\documentclass[11pt]{article}
\usepackage{amsthm,amssymb,amsfonts,amsmath,amscd}
\usepackage{nccmath}
\usepackage{graphicx}
\usepackage[pdfa]{hyperref}

\makeatletter \@addtoreset{equation}{section} \makeatother

\newtheorem{theorem}{Theorem}[section]
\newtheorem{lemma}{Lemma}[section]

\newtheorem{proposition}{Proposition}[section]

\newcommand{\E}{\mathbf{E}}
\newcommand{\mdet}{\mathrm{det}}

\newcommand{\Tr}{\mathrm{Tr}\,}

\begin{document}

\title{Characteristic polynomials of non-Hermitian random band matrices}
\author{ Mariya Shcherbina
\thanks{Institute of Science and Technology Austria and Institute for Low Temperature Physics, Kharkiv, Ukraine, e-mail: shcherbi@ilt.kharkov.ua. The
researches were supported by the ERC Advanced Grant "RMTBeyond" No. 101020331} \and
 Tatyana Shcherbina
\thanks{ Department of Mathematics, University of Wisconsin - Madison, USA, e-mail: tshcherbyna@wisc.edu. This material is based upon work supported  in part by Alfred P. Sloan Foundation grant FG-2022-18916 and the National Science Foundation  grant DMS-2346379}
}
\date{}
\maketitle

\begin{abstract}
We consider the asymptotic local behavior of the second correlation functions of the characteristic polynomials of a certain class of
 Gaussian $N\times N$ non-Hermitian random band matrices with a bandwidth $W$. Given $W,N\to\infty$, we show that this behavior near the point in the bulk of the spectrum
 exhibits the crossover at $W\sim \sqrt{N}$:  it coincides with those for Ginibre ensemble for $W\gg \sqrt{N}$, and
factorized as $1\ll W\ll \sqrt{N}$.  The result is the first step toward the proof of Anderson's type transition for non-Hermitian random band matrices.
\end{abstract}

\section{Introduction}\label{s:1}
We consider non-Hermitian random band matrices (RBM), i.e $N\times N$ matrices $H_N$
  whose entries $H_{ij}$ are independent random
complex variables with mean zero and variance determined by the so-called 
{\it band profile} $J$. This means
\begin{equation}\label{ban}
\mathbf{E}\big\{ H_{jk}\bar H_{jk}\big\}=J_{jk}
\end{equation}
with $J_{jk}$ taken to be small when $|j-k|\gg W$. The parameter $W$ is called the {\it bandwidth} of $H_N$.

In this paper we assume that $\{H_{ij}\}$ have Gaussian distribution and take 
\begin{equation}\label{J}
J=\left(-W^2\Delta+1\right)^{-1}
\end{equation}
with $\Delta$ being the discrete Laplacian on $[1,N]\cap \mathbb{Z}$ with Neumann boundary conditions:
\begin{equation*}
(-\Delta f)_j=\left\{\begin{array}{ll}
f_1-f_2,& j=1;\\
2f_j -f_{j-1}-f_{j+1},& j=2,\ldots, N-1;\\
f_n-f_{n-1},& j=N.
 \end{array}
 \right.
\end{equation*}
It is easy to see that  $J_{jk}\approx C_1W^{-1}\exp\{-C_2|j-k|/W\}$,  so it is exponentially small when $|j-k|\gg W$, as $W\to \infty$. Thus
matrices $H_N$ indeed can be considered as a special case of non-Hermitian random band matrices with the bandwidth $W$.  

It is easy to see that the probability law of  $H_N$ can be written in the form
\begin{equation}\label{band}
P_N(d H_N)=\prod\limits_{j,k=1}^N\dfrac{dH_{jk}d\overline{H}_{jk}}{\pi J_{jk}}e^{-\frac{|H_{jk}|^2}{J_{jk}}}.
\end{equation}

The Hermitian analog of matrices (\ref{ban}) plays an important role in mathematical physics.  Having nonzero entries only  in the strip of width $W$ around the main diagonal, Hermitian RBM  provide a natural model to study eigenvalue statistics and quantum transport in disordered systems as they interpolate between classical Wigner matrices, i.e. Hermitian random matrices with iid elements, and random Schr$\ddot{\hbox{o}}$dinger operators, where the randomness only appears in the diagonal potential. 
In particular,  Hermitian RBM can be used as a prototype of the celebrated  Anderson metal-insulator phase transition even in dimension one: for $W\gg \sqrt{N}$ the eigenvectors are delocalized and the eigenvalues have universal GUE local statistics, while the localized eigenvectors and Poisson statistics occurs for $W\ll \sqrt{N}$ (see \cite{FM:91}).
The recent  mathematical results justifying this conjecture for the Hermitian RBM in the dimension one and higher can be found in \cite{SS:Un}, \cite{CPSS:band_loc},   \cite{YY:2d_band}, \cite{YY:1d_band},  \cite{YY:3d_band}, \cite{Dr:band}, \cite{ErR:band} and references therein.

Despite the recent progress in studying universality of the local eigenvalue statistics for non-Hermitian matrices with iid entries  (see \cite{Ta-Vu:15}, \cite{CiErS:ed},\cite{MO:24}, \cite{O:24}, \cite{DY:24}, \cite{ASS:24} and references therein), the eigenvalue statistics of non-Hermitian matrices with a non-trivial spatial structure  
is much less accessible. In particular, for the non-Hermitian RBM (\ref{ban}) even justification of the expected convergence of the empirical spectral distribution  to the circular law, i.e. to the uniform distribution on a unit disk appearing as a limiting distribution for the non-Hermitian matrices with iid entries (see \cite{TV:08}, \cite{TVKr:10} and references therein), is  a highly non-trivial task. The best recent result \cite{H:25} shows this (weak) convergence only for non-Hermitian RBM with $W\gg N^{1/2+c}$ (see also \cite{H:24}, \cite{JJLO:21}, \cite{Tikh:23} and references therein for previous results).

In this paper we are going to study another spectral characteristic of the non-Hermitian RBM (\ref{ban}) -- (\ref{band}), namely,
the correlation functions of characteristic polynomials defined as
\begin{equation}\label{Theta_k}
\Theta_k (z_1,\ldots, z_k)=\mathbf{E}\Big\{\prod\limits_{s=1}^{k}\det(X_n-z_s)\det(X_n-z_s)^*\Big\},
\end{equation}
where expectation is taken with respect to (\ref{band}).

More precisely, we are interested in the asymptotic behavior of $\Theta_2$ for matrices (\ref{ban}) -- (\ref{band}), as $W, N\to\infty$,
and
\begin{align}\label{z_1,2}
z_1=z+\zeta/N^{1/2},\, z_2=z-\zeta/N^{1/2}, \quad |z|<1,
\end{align}
with $\zeta$ varying in a compact set in $\mathbb{C}$. 
To simplify the notations, we are going to drop the index $2$ in $\Theta_2$ below.

The interest to the characteristic polynomials of random matrices is stimulated by its connections to the
number theory, quantum chaos, integrable systems, combinatorics, representation
theory and others. In additional, although $\Theta_k$ is not a local object in terms of eigenvalue statistics, it is also expected 
to be universal in a certain sense. In particular, it was proved in \cite{Af:19} (see also \cite{Ak-Ve:03} for  the Gaussian (Ginibre) case)
that for non-Hermitian random matrices $H$ with iid complex entries with mean zero, variance one, and $2k$ finite moments
for any $z_j=z+\zeta_j/\sqrt N$, $j=1,..,k$ and $|z|<1$ we get
\begin{equation}\label{ChP_lim}
\lim\limits_{N\to\infty} N^{-\tfrac{k^2-k}{2}}\dfrac{\Theta_{k}(z_1,\ldots, z_k)}{\prod_j \Theta^{1/2}(z_j,z_j)}=C_k \dfrac{\det(K(\zeta_i,\zeta_j))_{i,j}^k}{|\Delta(\zeta)|^2}.
\end{equation}
Here 
\begin{equation}\label{K_b}
K(w_1,w_2)=e^{-|w_1|^2/2-|w_2|^2/2+w_1\bar w_2},
\end{equation}
$\Delta(\zeta)$ is  a Vandermonde determinant of $\zeta_1,\ldots,\zeta_k$, and
$C_k$ is constant depending only on the fourth cumulant $\kappa_4=\E[|H_{11}|^4]-2$ of the elements distribution, but not on the higher moments.
In particular, this means that the local limiting behavior  (\ref{ChP_lim}) for non-Hermitian matrices with iid entries coincides with those for the Ginibre ensemble as soon as the first four moments of elements distribution are Gaussian, i.e. the local behavior of the correlation functions of characteristic polynomials also exhibits a certain form of universality. Similar results were obtained for many classical Hermitian random matrix ensembles
(see, e.g., \cite{Br-Hi:00}, \cite{Br-Hi:01}, \cite{St-Fy:03},\cite{TSh:ChW}, \cite{TSh:ChSC},\cite{Af:16}, etc.)

Notice that for the Hermitian (or real symmetric) analog of RBM the local behaviour of the correlation function of characteristic polynomials exhibits 
the crossover at $W\sim \sqrt{N}$ similar to the crossover in the local eigenvalue statistics: it coincides with those for GUE/GOE ensemble for $W\gg \sqrt{N}$, and
factorized (which means that the limit in the r.h.s. of (\ref{ChP_lim}) is equal to 1) as $1\ll W\ll \sqrt{N}$ (see  \cite{TSh:14}, \cite{SS:17}, \cite{TS:20}, \cite{TS:22}). 
The goal of the current paper is to establish a similar result for non-Hermitian RBM (\ref{ban}) -- (\ref{band}). 
The method we use is based on the  SUSY transfer matrix approach  developed in \cite{SS:17} for the Hermitian case.

The main results are the following two theorems corresponding to delocalized and localized regimes of RBM respectively:
\begin{theorem}\label{t:Gin} Given the band matrix of the form (\ref{band}) with
 $W^2\gg N\log^2 N$, $W\le N^{1-\varepsilon_0}$ with some fixed $\varepsilon_0>0$, and $z_1, z_2$ of (\ref{z_1,2}), we have
\begin{align}\label{lim_Gin}
\lim_{N\to\infty,\frac{W^2}{N\log^2 N}\to \infty}\frac{\Theta(z_1,z_2)}{\Theta^{1/2}(z_1,z_1)\Theta^{1/2}(z_2,z_2)}=\frac{1-e^{-4|\zeta|^2}}{4|\zeta|^2},
\end{align}
which coincides with the limit (\ref{ChP_lim}) (i.e. with Ginibre case).
\end{theorem}
\begin{theorem}\label{t:locGin} Given the band matrix of the form (\ref{band}) with $W>N^{\varepsilon_0}$  with any fixed $\varepsilon_0>0$ and $W^2\ll N/\log N$, and $z_{1}, z_2$ of  (\ref{z_1,2}), we have
\begin{align}\label{lim_locGin}
\lim_{N\to\infty,\frac{W^2\log N}{N}\to 0}\frac{\Theta(z_1,z_2)}{\Theta(z,z)}=1.
\end{align}
\end{theorem}
These theorems are the first important steps
towards the proof of bulk universality and  Anderson's type transition for the non-Hermitian RBM.

The main idea of the paper is to represent $\Theta(z_1,z_2)$ in the form (see Proposition \ref{prop:IR})
\begin{align}\label{idea}
\Theta(z_1,z_2)=(\mathcal{K}_{\zeta}^{N-1}g,g)=\sum_{j=0}^\infty\lambda_{j}^{N-1}(\mathcal{K}_{\zeta})\psi_j(g),
\end{align}
where $\mathcal{K}_{\zeta}$ is an integral operator on the space of $2\times 2$ matrices, 
$|\lambda_{0}(\mathcal{K}_{\zeta})|\ge |\lambda_{1}(\mathcal{K}_{\zeta})|\ge \dots $ are its eigenvalues, and $\psi_j(g)$ are some scalar coefficients
which can be written in terms of right and left  eigenvectors corresponding to $\lambda_j(\mathcal{K}_\zeta)$. Of course,  $\lambda_j(\mathcal{K}_\zeta)$
and $\psi_j(g)$ depend on $W,N$. One can guess that if
\[
\Big|\frac{\lambda_{1}(\mathcal{K}_{\zeta})}{\lambda_{0}(\mathcal{K}_{\zeta})}\Big|\ll 1-C/N,
\]
then the main contribution to the sum in (\ref{idea}) comes from the term with $j=0$, and we can replace $\mathcal{K}_{\zeta}$ by its projection
on the eigenvector corresponding to $\lambda_{0}(\mathcal{K}_{\zeta})$. Thus, we obtain the result of Theorem \ref{t:locGin}. But if  we have an opposite inequality for the ratio
of two first eigenvalues, then many other terms in (\ref{idea}) may give a valuable contribution into the sum, and, therefore, one should expect the result of Theorem \ref{t:Gin}.
We will show below that
\[
\Big|\frac{\lambda_{1}(\mathcal{K}_{\zeta})}{\lambda_{0}(\mathcal{K}_{\zeta})}\Big|\sim 1-c/W^2,
\]
and, therefore,  the regime $W^2\ll N$ corresponds to Theorem \ref{t:locGin}, and  the regime $W^2\gg N$ gives the result of Theorem$\,$\ref{t:Gin}.

The paper is organized as follows. In Section 2 we use supersymmetry techniques  (SUSY) to derive the integral representation for $\Theta(z_1,z_2)$ and rewrite it as an action 
of the $N$-th degree of a transfer integral operator $K_\zeta$ on a space of $2\times 2$ complex matrices $Q$ (see (\ref{idea})).  Section 3 is devoted to the first step of the spectral analysis of $\mathcal{K}_{\zeta}$: we show
that the essential contribution to the sum (\ref{idea}) is given by the eigenvectors of $\mathcal{K}_{\zeta}$ concentrated in $W^{-1/2}\log W$-neighbourhood of 
 the ``maximum surface"  $Q=u_*U$ of the function (\ref{f}) (here $U$ is a $2\times 2$ unitary matrix, and $u_*=\sqrt{1-|z|^2}$), and so $\mathcal{K}_{\zeta}$ can be restricted to the neighbourhood of this surface by  changing $Q\to U(u_*+W^{-1/2}R)$ with $U\in U(2)$ and $R$ being a Hermitian $2\times 2$ matrix.  In Section 4 we perform a more detailed
spectral analysis of $\mathcal{K}_{\zeta}$ near the ``maximum surface" by considering separately the operator $\mathcal{A}_\zeta$ on the ``Hermitian part" $R$ and the operator $K_{R_1,R_2}$ on the ``unitary part" $U$ (see (\ref{cal_K})).
 
%In Section 4 we perform a more detailed
%spectral analysis of $\mathcal{K}_{\zeta}$ ,considering  first the action of $\mathcal{K}_{\zeta}$ on the vectors  $\Psi(R)$, then on the vectors $\Psi(R)h(U)$ and then on the vectors $\Psi(R-\epsilon\mathcal{M}(U))h(U)$ (see \ref{M(U)}). 

Section 5 is devoted to the proof of Theorems \ref{t:Gin}, \ref{t:locGin}. Some auxiliary results which we use in the proof  are proven in Appendix.

We denote by $C$, $C_1$, etc. various $W$ and $N$-independent quantities below, which
can be different in different formulas. To reduce the number of notations, we also use the same letters for the integral operators and their kernels.

\section{Integral representation}\label{s:2}
One can see that considering  $\tilde\Theta(z_1,z_2)= C\cdot \Theta(z_1,z_2)$ with any constant $C=C(N,W)$ does not change the limits (\ref{lim_Gin}) -- (\ref{lim_locGin}),
hence, for the future convenience below we consider the normalized version of $\Theta$:
\begin{align}\label{ti-theta}
\tilde\Theta(z_1,z_2)=(\pi^2W^2\lambda_*^{-1})^{2(N-1)} C_{N,W}^{-1}\,\Theta(z_1,z_2),
\end{align}
where  
\begin{align}\label{lambda_*}
 \lambda_*=1-W^{-1}(\alpha-u_*^2W^{-1}), \quad \alpha=u_*(2+u_*^2W^{-2})^{1/2},\quad u_*=(1-|z|^2)^{1/2},
\end{align}
and $C_{N,W}$ is defined below in (\ref{C}).

The main purpose of  this section is to obtain a convenient integral representation of $\tilde\Theta$ that can be rewritten in the operator form (\ref{idea}):
\begin{proposition}\label{prop:IR}
	Let  $H$ be the non-Hermitian Gaussian random band matrices defined by (\ref{ban}) -- (\ref{band}). Then 
	the normalized second correlation function of the characteristic polynomials  $\tilde\Theta$ defined by (\ref{ti-theta}) can be represented in the following form
\begin{align}\label{repr_Theta}
\tilde\Theta(z_1,z_2)=&\int\limits_{(H_2)^N} e^{f(Q_1)}(\prod_{j=1}^{N-1}\mathcal{K}_\zeta(Q_j,Q_{j+1}))e^{ f(Q_N)}
\prod\limits_{j=1}^n dQ_j=(\mathcal{K}_\zeta^{N-1}g,g),
\end{align}
where $H_2$ is the space of $2\times 2$ complex matrices, $\mathcal{H}=L_2(H_2)$, and $\mathcal{K}_\zeta: \mathcal{H}\to\mathcal{H}$ is an integral operator  with the kernel
\begin{align}
\mathcal{K}_\zeta(Q_j,Q_{j+1})=\pi^4W^4\lambda_*^{-2}&\exp\Big\{ -W^2 \Tr (Q_j- Q_{j+1})(Q_j- Q_{j+1})^*+f(Q_j)+f(Q_{j+1})\Big\},\label{K}
\end{align}
where
\begin{align}
 f(Q_j)=&\frac{1}{2}(-\Tr Q_jQ_j^*+\log \det \mathcal{Q}_j +2u_*^2),\quad g(Q)=e^{f(Q)},\label{f}\\
\mathcal{Q}_j =&\left(\begin{array}{cc}\hat z&iQ_j\\iQ_j^*&\hat z^*\end{array}\right),
%=\left(\begin{array}{cc} z+\zeta L&iQ_j\\iQ_j^*&\bar z+\bar\zeta L\end{array}\right)
\quad\hat z=\mathrm{diag}\{z_1,z_2\},
\label{Q_cal}\end{align}
and $\lambda_*$, $u_*$ are defined in (\ref{lambda_*}).
\end{proposition}
\begin{proof}
To derive the integral representation of $\Theta$ we will use  SUSY. The detailed information about the techniques and its applications 
to random matrix theory  can be found, e.g.,
in \cite{Ef} or \cite{M:00}.

Introduce vectors
\begin{align*}
&\Psi_l=(\psi_{jl})^t_{j=1,..,N}, \quad l=1,\ldots,4;\\
&\Psi_l^+=(\bar\psi_{jl})_{j=1,..,N}, \quad l=1,\ldots,4,
\end{align*}
with independent anticommuting Grassmann components $\{\psi_{jl}\}$, $\{\bar\psi_{jl}\}$.

Using the standard formula of Grassmann integration (see, e.g., \cite{Ef})
\begin{equation}\label{G_Gr}
\int \exp\Big\{\sum\limits_{j,k=1}^nA_{j,k}\overline{\chi}_j\chi_k\Big\}
\prod\limits_{j=1}^nd\,\overline{\chi}_jd\,\chi_j=\mdet A,
\end{equation}
we get
\begin{align*}
	\Theta(z_1,z_2) =  \E \Bigg\{
	\int \exp \Bigg\{ \sum\limits_{l = 1}^{2} \Psi^+_l\left(H_N - z_l\right)\Psi_l+ \sum\limits_{l = 1}^{2} \Psi^+_{l+2}\left(H_N - z_l\right)^*\Psi_{l+2} \Bigg\} d\Psi
	\Bigg\},
\end{align*}
where 
\begin{equation*}
   d\Psi =\prod\limits_{l=1}^4 \prod\limits_{j = 1}^{N} d\bar\psi_{jl}\,d\psi_{jl}.
\end{equation*}
Collecting the terms near $\Re H_{jk}$ and $\Im H_{jk}$, we can rewrite the formula as
\begin{align*}
	&\Theta(z_1,z_2) =  \int\exp \Big\{ -\sum\limits_{l = 1}^{2} z_l\Psi^+_l\Psi_l- \sum\limits_{l = 1}^{2} \bar z_l\Psi^+_{l+2}\Psi_{l+2} \Big\}\\
&\times	\E \Big\{
	\exp \Big\{ \sum\limits_{j,k = 1}^{N} \Re H_{jk}(\chi_{jk}^{(12)}+\chi_{kj}^{(34)})+ i\sum\limits_{j,k = 1}^{N} \Im H_{jk}(\chi_{jk}^{(12)}-\chi_{kj}^{(34)}) \Big\} \Big\}d\Psi
\end{align*}
with
\begin{align}\label{chi}
&\chi_{jk}^{(12)}=\bar\psi_{j1}\psi_{k1}+\bar\psi_{j2}\psi_{k2},\\ \notag
&\chi_{jk}^{(34)}=\bar\psi_{j3}\psi_{k3}+\bar\psi_{j4}\psi_{k4}.
\end{align}
After taking the expectation with respect to (\ref{band}), it gives
\begin{align*}
&\Theta(z_1,z_2) =  \int\exp \Big\{ -\sum\limits_{l = 1}^{2} z_l\Psi^+_l\Psi_l- \sum\limits_{l = 1}^{2} \bar z_l\Psi^+_{l+2}\Psi_{l+2} \Big\} \exp \Big\{ \sum\limits_{j,k = 1}^{N} J_{jk} \,\chi_{jk}^{(12)} \chi_{kj}^{(34)}\Big\}d\Psi.
\end{align*}
Applying Hubbard-Stratonovich transformation (see \cite{Ef})
\begin{align*}
e^{ab}=\pi^{-1} \int e^{a \bar u+b u-\bar u u} d\bar u\, du
\end{align*}
for $a,b$ being any commuting elements of Grassmann algebra, we get
\begin{align*}
&\Theta(z_1,z_2) = C'_{N,W}\int\exp \Big\{ -\sum\limits_{l = 1}^{2} z_l\Psi^+_l\Psi_l- \sum\limits_{l = 1}^{2} \bar z_l\Psi^+_{l+2}\Psi_{l+2} \Big\}\cdot \exp \Big\{ -\sum\limits_{j,k = 1}^{N} (J^{-1})_{jk}\Tr Q_j Q_k^*\Big\}\\
&\times \exp \Big\{-i\sum\limits_{j=1}^N(\bar\psi_{j1},\bar\psi_{j2})Q_j\begin{pmatrix}\psi_{j3} \\ \psi_{j4}\end{pmatrix}
-i\sum\limits_{j=1}^N(\bar\psi_{j3},\bar\psi_{j4})Q_j^*\begin{pmatrix}\psi_{j1} \\ \psi_{j2}\end{pmatrix}\Big\}d\Psi\, dQ,
\end{align*}
where $\{Q_j\}$ are complex $2\times 2$ matrices with independent entries and
\begin{equation}\label{dQ}
dQ=\prod\limits_{j=1}^N\prod_{p,r=1}^2\,d (\bar{Q}_j)_{pr} d(Q_j)_{pr}, \quad C'_{N,W}=\pi^{-4N} \mdet^{-4} J.
\end{equation}
The integral over $d\Psi$ can be taken now using (\ref{G_Gr}), and we obtain finally
\begin{multline*}
\Theta(z_1,z_2) \\=C'_{N,W}  \int\exp \Big\{-W^2\sum\limits_{j=1}^{N-1} \Tr (Q_j - Q_{j+1})(Q_j - Q_{j+1})^*-\sum\limits_{j=1}^{N} \Tr Q_j Q_j^*\Big\}\prod\limits_{j=1}^N \mdet \mathcal{Q}_j \,\,dQ,
\end{multline*}
with $\mathcal{Q}_j$ of (\ref{Q_cal}). Changing
\begin{equation}\label{C}
C_{N,W}=e^{-2Nu_*^2} \cdot C'_{N,W},
\end{equation}
we get (\ref{repr_Theta}).
\end{proof}

\section{Concentration of eigenfunctions of $\mathcal{K}_\zeta$}\label{s:3}

It is easy to see that for $\zeta=0$ the function $f$ of (\ref{f}) takes its maximum at $Q=u_*U$, with some unitary $U$ and $u_*$ of (\ref{lambda_*}).
Indeed, writing $Q=V_1\Lambda V_2$ with unitary $V_1,V_2$ and  $\Lambda=\mathrm{diag}\{\lambda_1,\lambda_2\}$, we have for $|z|<1$
\begin{align*}
f(\Lambda)=\frac{1}{2}\sum_{\alpha=1,2}(\log(|z|^2+\lambda_\alpha^2)-\lambda_\alpha^2+u_*^2)\le 0,
\end{align*}
and  the r.h.s. is zero iff $\lambda_1=\lambda_2=u_*$.

The aim of this section is to prove that the main contribution to (\ref{repr_Theta}) is given  by  $\cap_j \{\Omega_W(Q_j)\}$ with 
\begin{align}\label{Omega}
\Omega_W=\{Q:\|Q^*Q-u_*^2I_2\|\le \log W/W^{1/2}\}.
\end{align}
But before we would like to make the following observation on $\tilde\Theta(z_1,z_2)$.  It is evident from  (\ref{f})  that $\tilde\Theta(z+\zeta,z-\zeta)$ is an analytic function with respect to $\zeta$ and $\bar\zeta$ considered as independent variables.
Consider $\tilde\Theta(z_1,z_2)$   for $\zeta=\xi e^{i\phi}$, $\bar \zeta=\eta e^{-i\phi}$,  $\phi=\mathrm{arg}\, z$.
By the Cauchy-Schwartz  inequality, for any $w_1, w_2, w_3, w_4$
\begin{align*}
&\Big|\mathbf{E}\Big\{\prod_{j=1,3}\det (H-w_j)\det (H^*-\bar w_{j+1})\Big\}\Big|
\le \prod_{j=1}^4\mathbf{E}^{1/4}\Big\{|\det (H-w_j)|^4\Big\}=\prod_{j=1}^4\Theta^{1/4}(w_j,w_j).
\end{align*}
Applying this inequality to
\[w_1=z+\frac{\xi e^{i\phi}}{\sqrt{N}},\, w_2= z+\frac{\bar\eta e^{i\phi}}{\sqrt{N}},\,
w_3=z-\frac{\xi e^{i\phi}}{\sqrt{N}},\,w_4= z-\frac{\bar\eta e^{i\phi}}{\sqrt{N}},
\]
and using that
\[C_{n,W}(\pi^2W^2\lambda_*^{-1})^{-2(N-1)}\Big|\tilde\Theta(z_1,z_2)\Big|_{\zeta=\xi e^{i\phi},\bar \zeta=\eta e^{-i\phi}}=
\Big|\mathbf{E}\Big\{\prod_{j=1,3}\det (H-w_j)\det (H^*-\bar w_{j+1})\Big\}\Big|,
\]
we get that boundedness of $\big|\tilde\Theta(z_1,z_2)\big|$ for $\zeta=\xi e^{i\phi},\bar \zeta=\eta e^{-i\phi}$ follows from
the boundedness of $\tilde\Theta(z+\zeta/\sqrt N,z+\zeta/\sqrt N)$ for any $|\zeta|\le C$. Hence,
 by the uniqueness theorem, it is sufficient to prove the existence of the limit, as $N,W\to\infty$ of
 $\tilde\Theta(z_1,z_2)$ for $\zeta=\xi e^{i\phi},\bar \zeta=\eta e^{-i\phi}$, $\xi,\eta\in \mathbb{R}$.
 Notice that  by (\ref{f}) 
$\det \mathcal{Q}_j\in\mathbb{R}$, if  $\xi,\eta \in\mathbb{R}$. Thus, starting from this moment, we consider $\mathcal{K}_\zeta$ of (\ref{K}) as a positive operator
while for simplicity keeping notations $\mathcal{K}_\zeta$, $\zeta,\bar\zeta$.

\medskip

Recall the notation $\mathcal{H}=L_2(\mathbb{C}^4)$, where we consider $\mathbb{C}^4$ as a space of all $2\times 2$ matrices
with complex entries. Let $\mathbb{P}_W=\mathbf{1}_{\Omega_W}$ be the orthogonal projection in $\mathcal{H}$ on 
functions whose support lies in the domain $\Omega_W$ of (\ref{Omega}).
\begin{lemma}\label{l:conc} There is $N,W$-independent $C_1$ such that
\begin{align}\label{Om1}
\|(1-\mathbb{P}_W)\mathcal{K}_\zeta(1-\mathbb{P}_W)\|\le 1-C_1 \log W/W.
\end{align}
\end{lemma}
\textit{Proof.} Take $h\in(1-\mathbb{P}_W)\mathcal{H}$,  $\|h\|=1$. Choose $\delta=2u_*^2/3$, and let $h_\delta$ be a projection
of $h$ on $\Omega_{\delta}$, where $\Omega_{\delta}$ is defined by (\ref{Omega}) with $\log W/W^{1/2}$ replaced by $\delta$.
Then
\begin{align}\label{l:3.5}
(\mathcal{K}_\zeta h,h)\le 1-C_\delta(1-(\mathcal{K}_\zeta h_\delta,h_\delta)).
\end{align}
The inequality was proved in \cite{SS:17}(see Lemma 3.5), but for the reader's convenience we repeat its proof at the end of the proof of
Lemma \ref{l:conc}.

Consider the change of variables $Q_i= V^{(1)}_i\Lambda_iV^{(2)}_i$, where $V^{(1)}_i, V^{(2)}_i$ are unitary matrices and 
$\Lambda_i=\hbox{diag}\{\mu_{i1},\mu_{i2}\}$ ($\mu_{i1},\mu_{i2}>0$). The Jacobian of such change (see, e.g., \cite{Hu:63})  is
\[
\mathcal{J}(\Lambda)=4\pi^4(\mu_{i1}^2-\mu_{i2}^2)^2\det \Lambda_i.
\]
Then for function $h$ depending only on $\Lambda$ we have
\[\|h\|=\|\mathcal{J}^{1/2}h\|_{L_2(\mathbb{R}_+^2)}.\]
Write
\begin{equation*}
-W^2 \Tr (Q_1- Q_{2})(Q_1- Q_{2})^*
%\\
%=
= -W^2 \Tr (\Lambda_1^2+\Lambda_{2}^2)+\tilde k_\Lambda(V^{(1)*}_2 V^{(1)}_1,
V^{(2)}_1 V^{(2)*}_2)
\end{equation*}
with
\begin{align*}
\tilde k_\Lambda(V^{(1)*}_2 V^{(1)}_1,
V^{(2)}_1 V^{(2)*}_2)=W^2 \Tr(V^{(1)}_1\Lambda_1V^{(2)}_1(V^{(1)}_2\Lambda_2V^{(2)}_2)^*+(V^{(1)}_1\Lambda_1V^{(2)}_1)^*V^{(1)}_2\Lambda_2V^{(2)}_2).
\end{align*}
According to \cite{SW:03}, we have uniformly in $\Lambda_1^2,\Lambda_2^2 >u_*^2/3$ (i.e. for $Q_1,Q_2\in \Omega_\delta$)
\begin{align}\notag
\int dV^{(1)}dV^{(2)}&\exp\{ \tilde k_\Lambda(V^{(1)},V^{(2)})\}=\int dV^{(1)}dV^{(2)}\exp\{ W^2\Tr V^{(1)}\Lambda_1V^{(2)}\Lambda_2+cc\}\\
=&C\,\frac{\det \{I_0(2W^2\mu_{1i}\mu_{2j})\}_{i,j=1,2}}{W^4(\mu_{11}^2-\mu_{12}^2)(\mu_{21}^2-\mu_{22}^2) }\notag\\=
&C'\,\frac{\det \{e^{2W^2\mu_{1i}\mu_{2j}}\}_{i,j=1,2}}
{W^4(\mu_{11}^2-\mu_{12}^2)(\mu_{21}^2-\mu_{22}^2)
 (\det\Lambda_1\det\Lambda_2)^{1/2}}(1+O(W^{-2})),
\label{2unit}\end{align} 
where here and below ``cc" means the complex conjugate of the previous expression.
Here $I_0(z)$ is a modified Bessel function and we used the asymptotic relation
\[ I_0(z)=e^z\sqrt{\frac{2\pi}{z}}(1+O(z^{-1})).
\]
For an arbitrary  function  $\tilde f(Q)$ which depends only on ``eigenvalue part"  $\Lambda$ of $Q$ consider the operators:
\begin{align*}
%(\mathcal{K}_{\tilde f}h,h)=&(A_{\tilde f} \mathcal{J}^{1/2}h,\mathcal{J}^{1/2}h)_{L_2(\mathbb{R}_+^2)}+O(W^{-2})\|h\|,\\
\mathcal{K}_{\tilde f}(Q_1,Q_2)=&C_1W^{8}\exp\{-W^2 \Tr (Q_1- Q_{2})(Q_1- Q_{2})^*+\tilde f(Q_1)+\tilde f(Q_2)\},\\
A_{\tilde f}(\Lambda_1,\Lambda_2)=&C_2W^{4}\exp\{-W^2\Tr (\Lambda_1-\Lambda_2)^2+\tilde f(\Lambda_1)+\tilde f(\Lambda_2)\}.
\end{align*}
The above change of variables and (\ref{2unit}) imply
\begin{align*}
(\mathcal{K}_{\tilde f}h,h)=&(A_{\tilde f} \mathcal{J}^{1/2}h,\mathcal{J}^{1/2}h)_{L_2(\mathbb{R}_+^2)}+O(W^{-2})\|h\|.
\end{align*}
It's easy to see that there exist some absolute $c_*,d_*$ such that for $Q\in \Omega_\delta$ and $f$ of (\ref{f}) we have
\[f(Q)\le -c_*\Tr(\Lambda-u_*I_2)^2+d_*/N=:\tilde f(\Lambda).
\]
Consider 
\begin{align*}
h_{*\delta}(\Lambda)=\Big(\int |h_\delta(\Lambda,U,V)|^2dUdV\Big)^{1/2},\quad \|h_{*\delta}\|^2=\|h_{\delta}\|^2.
\end{align*}
Denote by $\tilde\psi_{\bar k}(\mu_1,\mu_2)$, $\bar k=(k_1,k_2)$  the eigenfunctions of $A_{\tilde f}$. Then,
similarly to Lemma \ref{l:A} below, we have
\[
\tilde\psi_{\bar k}(\Lambda)=W\kappa_{\bar k}H_{k_1}((W\tilde \alpha)^{1/2}\mu_1)H_{k_2}((W\tilde \alpha)^{1/2}\mu_2)e^{-W\tilde\alpha\Tr\Lambda^2},
\]
where $H_k$ is the  $k$th Hermite polynomial,  $\tilde\alpha=\sqrt{2c_*}(1+O(W^{-1}))$, and $\kappa_{\bar k}$ is the normalizing factor. Now
\begin{align*}
(\mathcal{K}_\zeta h_\delta,h_\delta)&\le (\mathcal{K}_{\tilde f} h_{\delta},h_{\delta})\le (\mathcal{K}_{\tilde f} h_{*\delta}, h_{*\delta})=
 (A_{\tilde f} \mathcal{J}^{1/2}h_{*\delta}, \mathcal{J}^{1/2}h_{*\delta})_{L_2(\mathbb{R}_+^2)}(1+O(W^{-2}))\\
=&\sum_{\bar k}\tilde\lambda_{\bar k}|( \mathcal{J}^{1/2}h_{*\delta},\tilde\psi_{\bar k}))_{L_2(\mathbb{R}_+^2)}|^2 (1+O(W^{-2})),
\end{align*}
where $\tilde\lambda_{\bar k}$ is the  eigenvalue of $A_{\tilde  f}$ corresponding to $\tilde\psi_{\bar k}$.
But, since $ \mathcal{J}^{1/2} h\in(1-\mathbb{P}_W)\mathcal{H}$, and $\|(1-\mathbb{P}_W)\tilde\psi_{\bar k}\|_{L_2(\mathbb{R}_+^2)}\le e^{-c\log^2W}$
 for $\max\{k_1,k_2\}<\log W$, we have 
\[
 ((1-\mathbb{P}_W)\mathcal{J}^{1/2}h_{*\delta},\tilde\psi_{\bar k}(\Lambda))_{L_2(\mathbb{R}_+^2)}=
(\mathcal{J}^{1/2}h_{*\delta},(1-\mathbb{P}_W)\tilde\psi_{\bar k}(\Lambda))_{L_2(\mathbb{R}_+^2)}\le e^{-c\log^2W}.
\]
Hence, in view of  the spectral theorem for $A_{\tilde f}$, we get
\begin{align*}
 &(A_{\tilde f} \mathcal{J}^{1/2}h_{*\delta}, \mathcal{J}^{1/2}h_{*\delta})_{L_2(\mathbb{R}_+^2)}
=\sum_{\bar k}\tilde\lambda_{\bar k}|( \mathcal{J}^{1/2}h_{*\delta},\tilde\psi_k)_{L_2(\mathbb{R}_+^2)}|^2\\ 
\le &
\sum_{\max\{k_1,k_2\}>\log W/2}\tilde\lambda_{\bar k}|(\mathcal{J}^{1/2}h_{*\delta},\tilde\psi_{\bar k}(\Lambda))_{L_2(\mathbb{R}_+^2)}|^2+O(e^{-c\log^2W})\\ \le&
\max_{\max\{k_1,k_2\}>\log W/2}\{\tilde\lambda_{\bar k}\}\|\mathcal{J}^{1/2}h_{*\delta}\|_{L_2(\mathbb{R}_+^2)}^2
=(1-C\log W/W)\|h_{*\delta}\|^2\le (1-C\log W/W).
\end{align*}
Using this bound in (\ref{l:3.5}), we obtain (\ref{Om1}).

Now let us prove (\ref{l:3.5}). Denote $\Omega_{\delta/2}$ the analogue of $\Omega_W$ of (\ref{Om1}) with $\log  W/W^{1/2}$ replaced by $\delta/2$
and set
\[ h_1=h\mathbf{1}_{\Omega_{\delta/2}},\quad h_2=h\mathbf{1}_{\Omega_\delta\setminus\Omega_{\delta/2}},\quad
h_3=h-h_1-h_2.
\]
Since $\mathcal{K}_\zeta \le \lambda_*^{-2}e^f$ and $\lambda_*=1+O(W^{-1})$, we have
\begin{align*}
&(\mathcal{K}_\zeta h,h)\le \lambda_*^{-2}(e^{f}h,h)\le \lambda_*^{-2}\|h_{1}\|^2+(1-C_{1\delta})\|h_2+h_3\|^2=1-C_{1\delta}\|h_2+h_3\|^2/2\\
\Rightarrow& \|h_2+h_3\|^2\le 2C_{1\delta}^{-1}(1-(\mathcal{K}_\zeta h,h)),
\end{align*}
where $C_{1\delta}=1-\max _{Q\not \in\Omega_\delta}e^{f(Q)}$.

Using the above bound and that  $(\mathcal{K}_\zeta h_1, h_3)=O(e^{-cW^2\delta})$ and $\|\mathcal{K}_\zeta\|\le\lambda_*^{-2}$,  we obtain
\begin{align*}
(\mathcal{K}_\zeta h,h)=&(\mathcal{K}_\zeta (h_1+h_2),h_1+h_2)+2\Re(\mathcal{K}_\zeta (h_1+h_2),h_3)
+(\mathcal{K}_\zeta h_3,h_3)\\
\le &(\mathcal{K}_\zeta (h_1+h_2),h_1+h_2)+2\lambda_*^{-2}\|h_2+h_3\|^2\\
\le &(\mathcal{K}_\zeta (h_1+h_2),h_1+h_2)+4 C_{1\delta}^{-1}\lambda_*^{-2}(1-(\mathcal{K}_\zeta h,h))\\
\Rightarrow&(\mathcal{K}_\zeta h,h)\le 1-(1+5 C_{1\delta}^{-1}\lambda_*^{-2})^{-1}\Big(1-(\mathcal{K}_\zeta (h_1+h_2),h_1+h_2)\Big).
\end{align*}
Since $h_1+h_3=h_\delta$, we get (\ref{l:3.5}).
$\square$

Now let us study $\mathbb{P}_W\mathcal{K}_\zeta\mathbb{P}_W$. 

Consider the cylinder  change of variables (see, e.g., \cite{Hu:63}) 
\begin{align}\label{c_ch}
Q_i=U_i\mathcal{R}_i, \quad U_i\in U(2), \quad \mathcal{R}_i>0,\quad J(\mathcal{R})=\pi^3(\Tr \mathcal{R})^2\det \mathcal{R}.
\end{align}
Everywhere below we consider our operators acting in $\mathcal{H}_0\otimes L_2(U(2))$ with
\begin{align}\label{H_0}
\mathcal{H}_0=L_2(\mathcal{H}_{2,+}) \,\text{with innner product}\, (\psi_1(\mathcal{R}),\psi_2(\mathcal{R}))=
\int_{\mathcal{H}_{2,+}}\psi_1(\mathcal{R})\overline{\psi_2(\mathcal{R})})d\mathcal{R}.
\end{align}
Here $\mathcal{H}_{2,+}$ is the space of all positive $2\times 2$ matrices and $d\mathcal{R}$ means the Lebesgue measure on $\mathcal{H}_{2,+}$.

Since $\mathbb{P}_W$ is the projector on $\Omega_W$ (see (\ref{Omega})),  Lemma \ref{l:conc} implies that we can restrict the integration with respect to $\mathcal{R}$ by $O(W^{-1/2}\log W)$-neighbourhood
 of $u_*I_2$, i.e. 
 \begin{align}\label{R}
 \mathcal{R}_i=u_*(I_2+W^{-1/2}R_i),\quad R_i=R_i^*,\quad\|R_i\|\le\log W+o(1).
 \end{align}
 Then we get
\begin{align}\label{Theta1}
\tilde\Theta(z_1,z_2)=&(\mathcal{K}_\zeta^{N-1}g,g), 
\end{align}
where $\mathcal{K}_\zeta$ is an integral operator with the kernel
\begin{align}\label{cal_K}
&\mathcal{K}_\zeta(R_1,U_1,R_2,U_2)= \mathcal{A}_\zeta(R_1,U_1,R_2,U_2) K_{R_1,R_2}(U_2^*U_1),
\\
\label{K,Z}
&K_{R_1,R_2}(U)=Z^{-1}(R_1,R_{2})e^{k_{R_1,R_2}(U^*_{2}U_1)},\\
\notag &k_{R_1,R_2}(U)=u_*^2W^2\Tr \Big((U-1)(1+R_1/W^{1/2})(1+R_{2}/W^{1/2})\Big)+cc,\\
& Z(R_1,R_{2})=(\pi u_*W)^{3}\int dU \exp\{k_{R_1,R_2}(U)\},\notag\\
  &\mathcal{Z}(R_1,R_2)=J^{1/2}(u_*(1+R_1/W^{1/2}))J^{1/2}(u_*(1+R_2/W^{1/2}))Z(R_1,R_{2}),
\label{cal_Z}\end{align}
with $J(\mathcal{R})$ defined in (\ref{c_ch}).
Operator $\mathcal{A}_\zeta$ of (\ref{cal_K}) has the form
\begin{align}\label{A}
\mathcal{A}_\zeta(R_1,U_1,R_2,U_2)=&e^{f_\zeta(R_1,U_1)}B(R_1-R_2)
e^{f_\zeta(R_1,U_1)}\mathcal{Z}(R_1,R_2)\\
f_\zeta(R,U)=&f(u_*U(1+R/W^{1/2})),\quad B(R)=(\lambda_*\pi u_*^2W)^{-2}e^{-Wu^2_*\Tr R^2},\notag
\end{align}
where $u_*,\lambda_*$ are defined in (\ref{lambda_*}), and $f$ is from (\ref{f}).

The function $g$ in (\ref{Theta1}) is obtained by the change of variables  (\ref{c_ch}) and (\ref{R}) in  $g$ of (\ref{f}):
\begin{align}\label{new_g}
g=e^{f_\zeta(R,U)},\quad \|g\|=CW(1+o(1)).
\end{align}
%where $f_\zeta(R,U)$ is given by (\ref{pert}) and (\ref{ti-f}).
Now let us expand $f_\zeta(R,U)$ around  $Q_*=u_*U$.  Introduce the block-diagonal unitary matrix $D(U)=\mathrm{diag}\{U,I\}$
and denote 
\begin{align}
L_{U^*}=U^*LU, \quad \epsilon=\big(W/N\big)^{1/2},\quad\mathcal{M}(U)=-\frac{1}{2u_*^2}(\zeta \bar zL_{U^*}+\bar\zeta  zL).
\label{M(U)}\end{align}
Notice that in the conditions of Theorems \ref{t:Gin} -- \ref{t:locGin} we have $\epsilon\le N^{-\varepsilon_0/2}$.

Then
\begin{align*}
&\hat Q_*^{-1}=D(U)\left(\begin{array}{cc} \bar z I_2&-iu_*\\-iu_*& z I_2\end{array}\right)D^*(U),\quad
\tilde Q=W^{-1/2}D(U)\left(\begin{array}{cc}\epsilon\zeta L_{U^*}&iu_*R\\i u_*R&\epsilon\bar\zeta L\end{array}\right)D^*(U),\quad \\
&\hat Q_*^{-1}\tilde Q=W^{-1/2}D(U)\left(\begin{array}{cc}\epsilon\bar z\zeta L_{U^*}+u_*^2R&i\bar z u_*R-i\epsilon u_*\bar\zeta L\\
izu_*R -i\epsilon u_*\zeta L_{U^*}&\epsilon\bar\zeta z L+u_*^2R.
\end{array}\right)D^*(U).
\end{align*}
Hence,
\begin{align}\label{pert}
f_\zeta(R,U)=&-\frac{u_*^2}{2W}\Tr R^2-\frac{1}{4}\Tr (\hat Q_*^{-1}\tilde Q)^2-\frac{1}{4}\sum_{p=3}^\infty\frac{(-1)^{p}}{p}\Tr (\hat Q_*^{-1}\tilde Q)^p\\
%-\frac{1}{2}\Tr (\hat Q_*^{-1}\tilde Q)^2&-\frac{u_*^2}{W}\Tr R^2
=&-\frac{u_*^2}{2W}\Tr (R-\epsilon\mathcal{M}(U))^2+N^{-1}\nu(U)+\tilde f_\zeta (R,U)+O(\epsilon^2 W^{-3/2}),\notag
\end{align}
where
\begin{align}\label{ti-f}
\tilde f_\zeta (R,U)=&W^{-3/2}\Tr R^3\varphi_0(1+R/W^{1/2})+(\epsilon/ W^{3/2})\Tr \mathcal{M}(U) R^2\varphi_1(1+R/W^{1/2}),\\
\nu(U)=&|\zeta|^2\Tr LU^*LU/2.
\label{nu}\end{align}
Here $\varphi_0(x)$ and $\varphi_1(x)$ are some $N,W$-independent function analytic around $x=1$, whose concrete form is not important for us.
We also denote  $\hat\nu(U)$ the operator of multiplication  by $\nu$.

Operator $\mathcal{A}_\zeta$ of  (\ref{A})  takes the form
\begin{align}\label{A_zeta}
&\mathcal{A}_\zeta(R_1, U_1,R_2,U_2)=F_\zeta(R_1, U_1)B(R_1-R_2)\mathcal{Z}(R_1,R_2)F_\zeta(R_2, U_2)\big(1+O(N^{-1}W^{-1/2})\big)
\\
%&\times(1+c_3\Tr R_1^3/W^{3/2}+c_3\Tr R_2^3/W^{3/2})(1+o(N^{-1})) \notag\\
&F_\zeta(R, U)= e^{-2u_*^4\Tr (R-\epsilon\mathcal{M}(U))^2/W+\nu(U)/N+\tilde f(R,U)},\qquad\quad F_0(R)=F_\zeta(R, U)\Big|_{\zeta=0}.
%\exp\Big\{\sum_{p\ge 3} c_p\Tr R^p/W^{p/2}\Big\}.
\notag
\end{align}
We will compare $\mathcal{A}_\zeta$ with operators
\begin{align}\label{hat_A}
\mathcal{A}(R_1,R_2)=&F_0(R_1)B(R_1-R_2)\mathcal{Z}(R_1,R_2)F_0(R_2)=\mathcal{A}_\zeta(R_1,R_2)\Big|_{\zeta=0},\\
\mathcal{A}_0(R_1,R_2)=&F_0(R_1)B(R_1-R_2)F_0(R_2)\label{A_0}\\
=&e^{W^{-3/2}\Tr R_1^3\varphi_0(1+R_1/W^{1/2})}\mathcal{A}_*(R_1,R_2)e^{W^{-3/2}\Tr R_2^3\varphi_0(1+R_2/W^{1/2})}\notag\end{align}
 with $\varphi_0$ of (\ref{ti-f}) and
\begin{align}
&\mathcal{A}_*(R_1,R_2)=e^{-u_*^4\Tr R_1^2/W}B(R_1-R_2)e^{-u_*^4\Tr R_2^2/W}\label{A_*},\\\
&\hskip1.8cm=\mathcal{A}_{*1}(x_{01},x_{02})\mathcal{A}_{*1}(x_{11},x_{12})\mathcal{A}_{*1}(x_{21},x_{22})\mathcal{A}_{*1}(x_{31},x_{32}),\notag\\
&\mathcal{A}_{*1}(x,y)=\Big(\frac{u_*^2W}{\pi\lambda_* }\Big)^{1/2}e^{-2u_*^4x^2/W}e^{-2Wu^2_*(x-y)^2}e^{-2u_*^4y^2/W}.
\notag\end{align}
In (\ref{A_*})  we  represent 
 \begin{align}\label{expR}
 R_{l}=x_{0l}I_2+x_{1l}\sigma_1+x_{2l}\sigma_2+x_{3l}\sigma_3,\quad l=1,2,
 \end{align}
where $\sigma_1,\sigma_2,\sigma_3$  are the Pauli matrices.

In the following lemma we compare the eigenvalues and eigenvectors of operator $\mathcal{A}$ of (\ref{hat_A}) with those of the ``quadratic form operator" $\mathcal{A}_*$
of (\ref{A_*}) (the last one can be computed explicitly via Hermite polynomials):
\begin{lemma}\label{l:A}
Let $\{\Psi_{*\bar m}(R),\lambda_{*\bar m}\}$  be eigenvectors and eigenvalues
of the operator $\mathcal{A}_*$ of (\ref{A_*}). Then
\begin{align}
 \Psi_{*\bar m}(R)=&P_{\bar m}(R)e^{-\alpha u_*^2\Tr R^2},\quad P_{\bar m}(R)=\prod_{i=0}^3H_{m_i}(u_*(2\alpha)^{1/2} x_i)/ \kappa_{m_i}, \label{psi_k}\\
 \lambda_{*\bar m}=&\lambda_{*}^{|m|},\quad
\bar m=(m_0,m_1,m_2,m_3),\, m_i=0,1,\dots,\quad |\bar m|=\sum_{i=0}^3|m_i|.
\notag\end{align}
Here $H_m(x)$ is the $m$'th Hermite polynomial, $\kappa_m$ is a normalization factor, and $\lambda_*$, $\alpha$ are defined in (\ref{lambda_*}).

Let $E_{|\bar m|}=\mathrm{Lin}\{\Psi_{*\bar j}\}_{|\bar j|=|\bar m|}$ and $\gamma(|\bar m|)=\mathrm{dim}E_{|\bar m|}$. Then  there are   $\gamma(|\bar m|)$
 eigenvalues $\{\lambda^{(\mu)}_{|\bar m|}\}_{\mu=1}^{\gamma(|\bar m|)}$ of $\mathcal{A}$ of (\ref{hat_A}) 
 such that
\begin{align}
&|\lambda^{(\mu)}_{|\bar m|}-\lambda_{*\bar m}|\le C(|\bar m|+1)W^{-2}. \label{diff_eig}\end{align}
If  an eigenvector $\Psi^{(\mu)}_{|\bar m|}(R)$ corresponds to  $\lambda^{(\mu)}_{|\bar m|}$, then for any  integer $p>0$ there are vectors 
\\ $\Psi^{(\mu)}_{*|\bar j|}\in E_{|\bar j|}$ such that
\begin{align}
&\Psi^{(\mu)}_{|\bar m|}(R)=\Psi_{*\bar m}^{(\mu)}(R)+\sum_{s=1}^{2p-1}\sum_{|\bar j-\bar m|\le s+2}W^{-s/2}\Psi^{(\mu)}_{*\bar j}(R)+O(W^{-p}).
\label{de_Psi.0}
\end{align}

Consider also  a ``deformed" operator 
\begin{align}\label{A_sh}
\mathcal{A}_{\mathcal{M}}(R_1,R_2)=&\big(1+(\epsilon/W)\Tr \mathcal{M}\phi(R_1)\big)\mathcal{A}(R_1,R_2)
\big(1+(\epsilon/W)\Tr \mathcal{M}\phi(R_2)\big)
%\\&+\mathcal{M}_1(R_1,R_2),\notag\\
% \mathcal{M}_1=&\sum_{L-q\le |\bar j|+|\bar k|\le L}m_{\bar j,\bar k}\Psi_{*\bar j}\otimes\Psi_{*\bar k},\quad \|\mathcal{M}_1\|\ll W^{-1}
%\mathcal{A}_{*M_1,M_2}(R_1,R_2)=&(\pi\lambda_* u_*^2W)^{-2}e^{-u_*^4\Tr (R_1-M_1)^2/W}e^{-Wu^2_*\Tr (R_1-R_{2})^2}e^{-u_*^4\Tr (R_1-M_2)^2/W},
%\label{M_1}
\end{align}
with $\epsilon=(W/N)^{1/2}$ and some analytic $\phi(R)$, and denote $\lambda_{\max}(\mathcal{A}_{\mathcal{M}})$ the maximum  eigenvalue of  
$\mathcal{A}_{\mathcal{M}}$. Then 
there is some   fixed  $k$ such that for any matrix 
$\mathcal{M}$ with $\|\mathcal{M}\|\le C$ (with an arbitrary absolute $C$) and $\Tr \mathcal{M}=0$ we have
\begin{align}\label{norm_A}
|\lambda_{\max}(\mathcal{A}_{\mathcal{M}})|\le \lambda_{\max}(\mathcal{A})(1+ k C^2/N),
%\|\mathcal{A}_{*M_1,M_2}\|\le e^{cW^{-1}\Tr(M_1-M_2)^2}
\end{align}
\end{lemma}
The proof of the lemma is given in Appendix.

Next we want to show that the main contribution to $\Theta(z_1,z_2)$ is given by the projection of $\mathcal K_\zeta$ on its first eigenvectors
concentrated in $\Omega_W$.

\section{Analysis of $\mathcal{A}$ and $K_{R_1,R_2}$}\label{s:4}
First we prove that $\mathcal{Z}(R_1,R_2)$ in operator $\mathcal A$ of (\ref{hat_A}) can be changed by $1$ with the small correction.
We also want to compare $\mathcal{Z}(R_1,R_2)$ with ``shifted" $\mathcal{Z}(R_1-\epsilon\mathcal{M},R_2-\epsilon\mathcal{M})$.
\begin{lemma}\label{l:Z_0} 

Given $\mathcal{Z}(R_1,R_2)$ of the form (\ref{cal_Z}) and
 $\Psi(R)\in \mathrm{Lin}\{\Psi_{*\bar k}\}_{|\bar k|\le m}$  ($m\ge 0$), we have
\begin{align}\label{act_D}
&\int \mathcal{A}_0(R_1,R_2)\Big(\mathcal{Z}(R_1,R_2)-1\Big) \Psi(R_2)dR_2=O((m+1)W^{-2}\|\Psi\|),
%\int \mathcal{A}_0(R_1,R_2)
%\Psi(R_2)dR_2+
%\\
%&\int \mathcal{A}_0(R_1,R_2)\mathcal{Z}(R_1,R_2) \Psi_{*\bar 0}(R_2)dR_2=\int \mathcal{A}_0(R_1,R_2)
%\Psi_{*\bar 0}(R_2)dR_2+O(W^{-5/2})
%\label{act_D.*}
\end{align}
where   $\mathcal{A}_0(R_1,R_2)$  was defined in  (\ref{A_0}).

In addition, for every fixed $2\times 2$ matrix $\mathcal{M}=\mathcal{M}^*$ and $\epsilon=(W/N)^{1/2}$
\begin{align}\label{diff_Z}
\int\mathcal{A}_0(R_1,R_2)\Big(\mathcal{Z}(R_1,R_2) -\mathcal{Z}(R_1-\epsilon\mathcal{M},R_2-\epsilon\mathcal{M})\Big)\Psi(R_2)dR_2
=O(\epsilon W^{-2}\|\Psi\|).
\end{align}

\end{lemma}
\textit{Proof}. Notice  that   $\mathcal{A}_0$ differs of $\mathcal{A}_*$ only by the  factors $(1+W^{-3/2}c_3\Tr R^3+O(W^{-2}))$ (see (\ref{A_*})).
In addition, if $\Psi(R)\in P_m\mathcal H$,  then $(1\pm c_3\Tr R^3/W^{3/2})\Psi(R)\in P_{m+3}\mathcal H$. Hence,
it is sufficient to prove (\ref{act_D}) for $\mathcal{A}_*$.

 We prove first that 
 %there exist some absolute constants $\kappa_1,\kappa_2,\kappa_3,\kappa_4$  such that\\
for $\|R_1-R_2\|\le W^{-1/2}\log W$ we have
\begin{align}\notag
&\mathcal{Z}(R_1,R_2)=1+\Delta(R_1,R_2),\\ 
\Delta(R_1,R_2)=&\frac{u_*^2}{2\Tr S}\Tr[R_2,R_1][R_1,R_2]-\frac{\Tr(R_1^\circ+R_2^\circ)^2/4}{W\Tr S}+O(W^{-2}),\label{Delta}\\
&S=\frac{1}{2}\{1+R_1/W^{1/2},1+R_2/W^{1/2}\}.
\label{S}\end{align}
Here and below for arbitrary matrices $A,B$ we use the notations
\begin{align}\label{A^circ}
 \{A,B\}=AB+BA,\quad
 [A,B]=AB-BA,\quad A^\circ=A-\frac{\Tr A}{2} I_2.
\end{align}
Indeed, in order to obtain (\ref{act_D}), we  need to integrate over $R_2$ the kernel $\mathcal{A}_*(R_1,R_2)\mathcal{Z}(R_1,R_2)$  multiplied by the function of the form
\[\Psi(R_2)=e^{-u_*^2\alpha\Tr R^2}p(R_2),\quad \mathrm{deg}\,p(R_2)\le m
\]
with $\alpha$ of (\ref{lambda_*}).
Complete the square at the exponent:
\begin{align}\notag
\mathcal{A}_*(R_1,R_2)e^{-u_*^2\alpha\Tr R_2^2}=&\Big(\frac{u_*^2W}{\pi\lambda_* }\Big)^2\exp\Big\{-Wu_*^2\Tr(R_1-R_2)^2-\alpha u_*^2\Tr R_2^2-u_*^4\Tr (R_1^2+R_2^2)/W\Big\}\\
=&\Big(\frac{u_*^2W}{\pi\lambda_* }\Big)^2\exp\Big\{-u_*^2(W+\alpha+ u_*^2/W)\Tr(R_2-\mu R_1)^2-C\Tr R_1^2\Big\},\label{repr_A}\\
\mu=&W/(W+\alpha+u_*^2/W)=1+O(W^{-1}).\notag
\end{align}
The constant $C$ here is not important since we integrate over $R_2$.  

Take $\Psi(R)=p(R)e^{-\alpha u_*^2\Tr R^2}$ with $p(R)$ being a polynomial of entries  of
$R$ of degree at most $m$. Given (\ref{Delta}), we substitute the r.h.s. of (\ref{Delta}) to the l.h.s. of (\ref{act_D}). 
Using (\ref{repr_A}), integrating by parts with respect to $R_2$, and taking 
into account that  the derivative  of $\Tr S(R_1,R_2)$ will give us an additional factor $W^{-1/2}$ (see (\ref{S})), we obtain for the first term of the r.h.s. of (\ref{Delta}):
\begin{align}\label{int_com}
&\frac{u_*^2}{2}\int \mathcal{A}_*(R_1,R_2)e^{-\alpha u_*^2\Tr R_2^2}\Tr[R_2-\mu R_1,R_1][R_1,R_2-\mu R_1]p(R_2)\Tr ^{-1}S(R_1,R_2)dR_2\\
&\quad =\frac{\Tr (R_1^\circ)^2}{W\Tr S(R_1,R_1)}\Psi(R_1)+O(\sqrt mW^{-2}\|\Psi\|)+O(W^{-5/2}\|\Psi\|).
\notag\end{align}
Here we used that for the normalized Hermite polynomial $(m!)^{-1/2}H_m(x)$ we have
\begin{align}\label{der_Herm}
(m!)^{-1/2}H_m'(x)=\sqrt m ((m-1)!)^{-1/2}H_{m-1}(x).
\end{align}
For the second term of the r.h.s. of (\ref{Delta}) we also get
\begin{align*}
\int \mathcal{A}_*(R_1,R_2)e^{-\alpha u_*^2\Tr R_2^2}\frac{\Tr(R_1^\circ+R_2^\circ)^2/4}{W\Tr S(R_1,R_2)}p(R_2)dR_2=
\frac{\Tr(R_1^\circ)^2}{W\Tr S(R_1,R_1)}\Psi(R_1)\\
+O(mW^{-2}\|\Psi\|)+O(W^{-5/2}\|\Psi\|),
\end{align*}
and so the integral with the r.h.s. of (\ref{Delta}) gives $O(mW^{-2}\|\Psi\|)$. This implies (\ref{act_D}).

Thus, we are left to prove (\ref{Delta}).
To simplify formulas below we set
\begin{align}\label{R,De}
R=(R_1+R_2)/2,\quad D= W^{1/2}(R_1-R_2)/2,\,\,\|D \|\le\log W ,\quad R_{1,2}=R\pm D /W^{1/2}.
\end{align}
Let us transform $k_{R_1,R_2}(U)$ of (\ref{K,Z}) into a more convenient form, using  notations (\ref{R,De})  and (\ref{A^circ}):
\begin{align*}\notag
k_{R_1,R_2}(U)=&(u_*W)^2\Tr((U+U^*)/2-1)\{1+R_1/W^{1/2},1+R_2/W^{1/2}\}\\
&+(u_*W)^2\Tr((U-U^*)/2)[R_1/W^{1/2},R_{2}/W^{1/2}]\\=&k_{*R_1,R_2}(U)+\rho_1+\rho_2,\notag
\end{align*}
where $S$ was defined in (\ref{S}), and
\begin{align}\label{k_*}
k_{*R_1,R_2}(U)=&(u_*W)^2\Tr S\,\Tr((U+U^*)/2-1)),\\
\label{rho}
\rho_1=&2u_*^2W^2\Tr \frac{(U+U^*)^\circ}{2} S^\circ,
%=u_*^2W^{3/2}\Tr (U+U^*)^\circ \Big(2R+R^2W^{-1/2}-D ^2W^{-3/2}\Big)^\circ \\
%&\rho_2=2u_*^2W^{3/2}\Tr (U+U^*)^\circ R,\notag\\
\qquad\qquad\rho_2=u_*^2W\Tr \frac{(U-U^*)^\circ}{2}[R_1 , R_2].
%=2u_*^2W^{1/2}\Tr (U-U^*)^\circ[D , R]\notag
\end{align}
Denote $\mathcal{T}(\phi)=\mathrm{diag}\{e^{i\phi/2},e^{-i\phi/2}\}$ and represent $U$ as
\begin{align}\label{U}
 U=&\mathcal{T}(\varphi)\left(
\begin{array}{cc}\cos(\theta/2)&i\sin(\theta/2)\\
i\sin(\theta/2)&\cos(\theta/2)\end{array}\right)\mathcal{T}(\psi)e^{i\gamma}\\=&
\left(\begin{array}{cc}
\cos(\theta/2)e^{i(\sigma+\gamma)}&i\sin(\theta/2)e^{i(\delta+\gamma)}
\\
i\sin(\theta/2)e^{i(-\delta +\gamma)}&\cos(\theta/2)e^{i(-\sigma+\gamma)}
\end{array}\right), \, 
 \sigma=\frac{1}{2}(\varphi+\psi),\, \delta =\frac{1}{2}(\varphi-\psi),
\notag\end{align}
where $\gamma\in [-\pi/2,\pi/2]$, $\sigma,\delta\in[-\pi,\pi]$, $\theta\in[0,\pi]$. Then
\begin{align}\notag
&\Tr((U+U^*)/2-1)=-2(1-\cos(\theta/2)\cos\sigma\cos\gamma),\\
&\frac{1}{2}(U+U^*)^\circ =-\sin\gamma\,\tilde U,\quad \tilde U=\left(
\begin{array}{cc}\cos(\theta /2)\sin\sigma&e^{i\delta }\sin(\theta/2) \\e^{-i\delta }\sin(\theta/2)
&-\cos(\theta/2)\sin\sigma\end{array}\right),\, 
\label{ti_U}\\
&\frac{1}{2}( U- U^*)^\circ=
i\cos\gamma\, \tilde U.
%\left(
%\begin{array}{cc}i\cos(\theta/2)
%\sin\sigma&ie^{i\delta }\sin(\theta/2)\\ie^{-i\delta }\sin(\theta/2)&-i\cos(\theta/2)\sin\sigma\end{array}\right).
\notag\end{align}
Analyzing the integral with respect to $\theta$, $\varphi$, $\psi$ and $\gamma$,  we conclude that the main contributions to the integral is given 
by the range of these variables where
\begin{align*}\notag
&\sin(\theta/2),\sin\sigma,\sin\gamma\sim W^{-1},\,\, (1-\cos(\theta/2)\cos\sigma\cos\gamma)\sim W^{-2}.
%\quad\sin\theta_V\sim (Wd_1d_2)^{-1/2} 
\end{align*}
Hence, using these relation and taking into account that $S^\circ\sim W^{-1/2}$ and\\ $[R_1,R_2]=[R_1-R_2,R_2]\sim W^{-1/2}$, we obtain
\begin{align}
&|\rho_1|\le W^{-1/2}\log W,\quad |\rho_2|\le W^{-1/2}\log W,
\label{b_rho}\end{align}
and, thus, we can expand $\exp\{k_{R_1,R_2}(U)\}$ with respect to $\rho_1,\rho_2$.

Set
\begin{align}\label{<f>}
\left\langle f\right\rangle=&Z_0^{-1}\int dU f(U)\exp\{-2u_*^2W^2\Tr S(1-\cos(\theta/2)\cos\sigma\cos\gamma)\},\\
Z_0=&\int dU \exp\{-2u_*^2W^2\Tr S(1-\cos(\theta/2)\cos\sigma\cos\gamma)\}\notag\\=&(2\pi u_*^2W^2\Tr S )^{-2}(1+O(W^{-2})).
\label{Z_0}\end{align}
Observe that
\begin{align}\notag
& \int dU f(U)\exp\{-2u_*^2W^2\Tr S(1-\cos(\theta/2)\cos\sigma\cos\gamma)\}\\
&=\frac{1}{(2\pi)^3}\int_0^\pi \sin\theta d\theta\int_{-\pi}^\pi  d\phi d\psi d\gamma f(U)\\
&\qquad\times\exp\Big\{2u_*^2W^2\Tr S\Big(\cos(\theta/2)-1+\cos\sigma-1+\cos\gamma-1\Big)\Big\}
(1+O(W^{-2}))\notag\\
&\Rightarrow \left\langle f_1(\theta)f_2(\sigma,\delta )f_3(\gamma)\right\rangle= \left\langle f_1(\theta)\right\rangle 
\left\langle f_2(\sigma,\delta)\right\rangle  \left\langle f_3(\gamma)\right\rangle +O(W^{-2}).
\label{fact}
\end{align}
which implies
\begin{align}\label{int}
&\left\langle \sin^2\gamma\right\rangle =(2u_*^2W^2\Tr S)^{-1}(1+O(W^{-2})),\\
&\left\langle\sin^2\sigma\right\rangle=(2u_*^2W^2\Tr S)^{-1}(1+O(W^{-2})),\notag\\
&\left\langle\sin^2(\theta/2)\right\rangle
=(u_*^2W^2\Tr S)^{-1}(1+O(W^{-2}))\notag\\
& \left\langle \sin^{2\alpha+1}\gamma\right\rangle = \left\langle \sin^{2\alpha+1}2\gamma\right\rangle= \left\langle \sin^{2\alpha+1}\sigma\right\rangle=\left\langle \sin^{2\alpha+1}2\sigma\right\rangle =0,\quad\alpha=0,1\notag\\
 & \left\langle e^{ ik\delta }\right\rangle\sim W^{-2},\,k=\pm1,\pm2,\dots.
%  \\
%\left\langle e^{im\gamma}\sin\gamma\sin\sigma\right\rangle\sim W^{-4}, \left\langle  e^{im\gamma}\sin\gamma\sin(\theta/2)\right\rangle\sim W^{-3},
% \left\langle e^{im\gamma}\sin\gamma e^{ik\delta }\right\rangle\sim W^{-4}
\notag\end{align}
Notice that by (\ref{int}), (\ref{fact}),  and (\ref{Z_0}) we have
\begin{align}\notag 
&\left\langle\rho_1^{\alpha}\rho_2^{s-\alpha}\right\rangle =O(W^{-2}),\quad (s\ge 3,\,0\le \alpha\le s)\\
\Rightarrow&Z=Z_0\left\langle(1+\rho_1^2/2+\rho_2^2/2+O(W^{-2}))\right\rangle, \notag\\
\Rightarrow &\mathcal{Z}=G(R_1,R_2)
\left\langle1+\rho_1^2/2+\rho_2^2/2\right\rangle (1+O(W^{-2}))\,\,\text{with}\label{calZxy}\\
&G(R_1,R_2)=J^{1/2}(u_*(1+W^{-1/2}R_1)J^{1/2}(u_*(1+W^{-1/2}R_2))(u_*^2\Tr S)^{-2}.
\notag\end{align}
Using (\ref{int}), (\ref{fact}) and  the form of $\tilde U$ (see (\ref{ti_U})), we get
\begin{align*}
\left\langle\Tr\tilde U A^\circ \Tr\tilde U B^\circ\right\rangle=(u_*^2W^2\Tr S)^{-1}\Tr (A^\circ B^\circ)(1+O(W^{-2})).
\end{align*}
Thus,
\begin{align}\label{rho^2}
\frac{1}{2}\left\langle\rho_1^2\right\rangle=&2(u_*^2W^2)^2\left\langle\sin^2\gamma \Tr^2\tilde US^\circ\right\rangle=\frac{\Tr(S^\circ)^2}{\Tr^2 S}
(1+O(W^{-2}))\\
%=\frac{4\Tr(R^\circ)^2(1+\Tr R/W)}{\Tr^2 S}\\
\frac{1}{2}\left\langle\rho_2^2\right\rangle=&(u_*^2W)^2\left\langle\cos^2\gamma \Tr^2 \tilde U (i[R_1,R_2])\right\rangle=\frac{u_*^2}{2\Tr S}
\Tr [R_1,R_2][R_2,R_1](1+O(W^{-2})).
\notag\end{align}
Then, introducing notations
\begin{align*}
&x=W^{-1/2}\Tr R,\quad y^2=W^{-1}\Tr (R^\circ)^2,
\end{align*}
 we obtain 
\begin{align*}\notag
&J^{1/2}(u_*(1+W^{-1/2}R_1)J^{1/2}(u_*(1+W^{-1/2}R_2))=J(u_*(1+W^{-1/2}R))+O(W^{-2})\\
&=4u_*^4(1+x+x^2/4)(1+x+x^2/4-y^2/2)+O(W^{-2}),\notag\\
&\Tr S=2(1+x+x^2/4+y^2/2)+O(W^{-2}),\\& \Tr (S^\circ)^2=4y^2(1+x)+O(W^{-2}),\notag\\
&G(R_1,R_2)=\frac{(1+x+x^2/4)(1+x+x^2/4-y^2/2)}{(1+x+x^2/4+y^2/2)^2}=1-\frac{3}{2}y^2(1+x/2)^{-2}+O(W^{-2}).
\end{align*}
The above relations combined with (\ref{calZxy}) and (\ref{rho^2}) finish the proof of (\ref{Delta}). 

To prove (\ref{diff_Z}), notice that it follows from the above that there are asymptotic expansion of coefficients of $\mathcal{Z}(R_1,R_2)$ with
respect to $W^{-1}$. Since these expansions starts from $W^{-2}$ and the coefficients depend on $R_1,R_2$ through traces of some polynomials of
$R_1,R_2$, we conclude that (\ref{diff_Z})
is true.

$\square$

Now we are going to study the ``unitary part" $K_{R_1,R_2}$ of  operator $\mathcal{K}_\zeta$ (see (\ref{K,Z})).
Consider some $M\gg 1$ and set
\begin{align}\label{E^ell}
\mathcal{E}^{(\ell)}=\hbox{Lin}\{t^{(\ell)}_{0p}(U)\}_{p=-\ell}^\ell,\quad
\mathcal{E}_{M}=
\cup_{\ell=0}^{M}\mathcal{E}^{(\ell)},\quad \mathcal{E}=\cup_{\ell} \mathcal{E}^{(\ell)}.
\end{align}
Here $\{t_{mp}^{(\ell)}(U)\}_{m,p=-\ell}^\ell$ are the coefficients of
 the $\ell$'th irreducible representation $T^{(\ell)}(U)$ of $SU(2)$.
We denote also by $\hat{\mathcal{E}}_{M}$ the orthogonal projection on $\mathcal{E}_{M}$, and by $\hat{\mathcal{E}}^{(\ell)}$ the orthogonal projection
on $\mathcal{E}^{(\ell)}$.

Since the function $g$ of  (\ref{f})  and $\mathcal{K}_\zeta$ depend on $U$ only  via $L_{U^*}$, they do not depend on $\det U$ and $\phi$ in (\ref{U}).
 Hence, $\mathcal{K}_\zeta$  can be considered as an operator acting in $\mathcal H_0\otimes\mathcal{E}$ (recall that $\mathcal{H}_0$
 means the $L^2$-space on all positive $2\times 2$ matrices). Moreover,
since  the kernel of $K_{R_1,R_2}$  depends on $U_1U_2^* $,  the operator commutes with 
all ``shift operators" and   each $\mathcal{E}^{(\ell)}$   reduces the operator $K_{R_1,R_2}(U_1U_2^*)$. 

\begin{lemma}\label{l:KPsi}
Given operator $K_{R_1,R_2}$ with a kernel (\ref{K,Z}), we have for $\|R_1-R_2\|\le W^{-1/2}\log W$ and $\ell\le W^{3/4}\log^2W$:
\begin{align}\label{Kt.0}
K_{R_1,R_2}t^{(\ell)}_{0k}=&\tilde\lambda_{\ell}t^{(\ell)}_{0k}
+b^{(\ell)}_{k+1}t^{(\ell)}_{0k+1}+b^{(\ell)}_{k}t^{(\ell)}_{0k-1}
+O(\ell^2\log^2W/W^3),\\
\tilde\lambda_{\ell}=&\lambda_\ell+O(\ell^2/W^2)\big(O(R_1/W^{1/2})+O(R_2/W^{1/2})\big),\notag\\
b^{(\ell)}_{k}=&d^{(\ell)}_{k}(\ell/W) [R_2-R_1,R_1]_{12}+O(\ell W^{-5/2}\log^3 W),
\notag
\end{align}
where $d^{(\ell)}_{k}$ are some bounded constants which are not important for us, and
\begin{equation}\label{lam_l}
\lambda_\ell=1-\ell(\ell+1)/8(u_*W)^{2}.
\end{equation}
Moreover, for any function   $\Psi_h(R,U)=\Psi(R)h(U)$
with  $\Psi(R)\in \mathcal {H}_L$ (see (\ref{P_L})) and  \\
  $h\in\mathcal{E}^{(\ell)},\,\ell\le cW^{3/4},\|h\|=1$ we have
\begin{align}
\Big\|(\mathcal{K}_0\Psi)(U_1,R_1)-\lambda_{\ell}h(U_1)( \mathcal{A}\Psi)(R_1)\Big\|
\le C\|\Psi\|(W^{-1/2}(\ell/W)^2+L \ell/W^2),
\label{KPsi.1}\end{align}
where $\mathcal{K}_0=\mathcal{K}_\zeta\Big|_{\zeta=0}$ and $\mathcal{A}$ was defined in (\ref{hat_A})
 If $\Psi_{0,h}(R,U)=\Psi_{ 0}(R)h(U)$, where $\Psi_0$ is an eigenvector of $\mathcal{A}$ corresponding to $\lambda_{\max}$
  and $h\in\mathcal{E}^{(\ell)},\,\ell\le cW^{3/4},\|h\|=1$, then
\begin{align}
\Big\|(\mathcal{K}_0\Psi_{0,h})-\lambda_\ell\lambda_{\max} \Psi_{0,h}\Big\|
= O(W^{-1/2}(\ell/W)^2).
\label{KPsi0}\end{align}
%where $ \lambda_\ell $ was defined in (\ref{Kt.0}).

In addition, for every fixed $2\times 2$ matrix $\mathcal{M}=\mathcal{M}^*$, $\epsilon=(W/N)^{1/2}$, 
$h\in\mathcal{E}_{M},\,M\le W^{1/2}/L$, and $\Psi \in \mathcal{H}_L$
\begin{align}\label{diff_K_eps}
\int\mathcal{A}(R_1,R_2) \Big((K_{R_1,R_2}h)(U_1)- (K_{R_1-\epsilon\mathcal{M},R_2-\epsilon\mathcal{M}}h)(U_1)\Big)&\Psi(R_2)dR_2\\
=&O(\epsilon W^{-2}ML\|\Psi\|),
\notag\end{align}
and for $\ell>W^{3/4}\log W$
\begin{align}\label{b_K_R}
\mathcal{E}^{(\ell)}K_{R_1,R_2}\mathcal{E}^{(\ell)}\le
1-CW^{-1/2}\log^2W
\end{align}

\end{lemma}

\textit{Proof.} 
Applying $K_{R_1,R_2}$ of (\ref{K,Z})  to $t_{0k}^{(\ell)}$ and changing the integration variable $U_2\to U_1U^*$, we obtain
%and then $U\to \tilde UV^*$, 
\begin{align}\label{Kt}
&({K}_{R_1,R_2}t_{0k}^{(\ell)},t_{0k'}^{(\ell)})=
Z^{-1}\int\exp\{k_{R_1,R_2}(U_2^*U_1)\}t_{0p}^{(\ell)}(U_2)dU_2\\ 
&%=\sum_{s}t_{0s}^{(\ell)}(U_1)\int\exp\{k_{R_1,R_2}( U)\}t_{sk}^{(\ell)}( U^*)dU
=\sum_{s}t_{0s}^{(\ell)}(U_1)\mathcal F_{sk}^{(\ell)}({R_1,R_2}),
\notag\end{align}
 where
\begin{align}\notag
&\mathcal{F}_{sk}^{(\ell)}(R_1,R_2)=Z^{-1}\int \exp\{k_{R_1,R_2}( U)\}t_{sk}^{(\ell)}( U^*)dU.
\end{align}
Then we need to analyse
\begin{align}\label{F_ks.1}
\mathcal{F}_{sk}^{(\ell)}(R_1,R_2)=&\frac{\left\langle t_{sk}^{(\ell)}(\tilde U^*)(1+\sum (\rho_1+\rho_2)^m/m!)\right\rangle}
{\left\langle 1+\sum  (\rho_1+\rho_2)^m/m!\right\rangle},
%\\=&\left\langle (1-c_{j}^{(\ell)}x)e^{2ij\phi_{\tilde U}}\right\rangle +O(W^{-5/2})=1-\frac{c_\ell}{W^2\sigma_1\sigma_2}+O(W^{-3}).
\end{align}
where $\rho_1, \rho_2$ are defined in (\ref{rho}), and $\langle\cdot\rangle$ is defined in (\ref{<f>}).

To prove (\ref{Kt.0})  we use  formulas (see \cite{Vil:68})
\begin{align}\label{as_P}
&t_{sk}^{(\ell)}(\tilde U)=P_{sk}^{(\ell)}(\cos\theta) e^{i(s\varphi+k\psi)}=P_{sk}^{(\ell)}(\cos\theta)e^{i(s+k)\sigma+i(s-k)\delta}, 
\end{align}
and the following proposition:
\begin{proposition}\label{p:as_Leg} If $|\sin(\theta/2)|\le W^{-1}\log W$ and $\ell<W^{3/4}\log^2W$,
then there is a constant $\kappa>0$ such that
\begin{align}\label{as_Leg.0} 
P_{k+1,k}^{(\ell)}(\cos\theta)=&i(1+(k+1)/\ell)^{1/2}(1-k/\ell)^{1/2}\ell \sin(\theta/2)(1+O(\ell \sin^2(\theta/2))\\
|P_{k+q,k}^{(\ell)}(\cos\theta)|\le& (\kappa\ell\sin(\theta/2))^q,\quad q\ge 2\notag\\
P^{(\ell)}_{00}(\cos\theta)=&1-\ell(\ell+1)\sin^2(\theta/2)+O((\ell\sin(\theta/2))^3).
\label{as_Leg.1}\end{align}
In addition, 
for any $\ell>W^{3/4}\log^2W$, if $\|R_1-R_2\|\le CW^{-1}\log W$ then
\begin{align}\label{as_Leg} 
\Big|\left\langle t_{00}^{(\ell)}(U)\right\rangle\Big|\le 1- CW^{-1/2}\log^4 W/2.
\end{align}
\end{proposition}
The proof of the proposition is  given in Appendix.

\medskip

If $|q|\ge 1$,  then since $k_{R_1,R_2}$ depends on $e^{i\delta}$ only via $\rho_1,\,\rho_2$ (see (\ref{rho})), the integration with respect to $\delta$ gives us an extra multiplier of the order $W^{-2}$ (see (\ref{int})) unless we integrate the terms
with $(\rho_1+\rho_2)^{q'}$ with $q'\ge|q|$. Thus, by  (\ref{rho}), (\ref{ti_U}), (\ref{b_rho}), and (\ref{as_Leg.0}), we have
\begin{align}\label{sum_F}
&\sum_{|q|\ge 2}|\mathcal{F}_{k+q,k}^{(\ell)}(R_1,R_2)|\le C(\ell/W)^2(\|R_1-R_2\|^2+W^{-1}),\\
&\mathcal{F}_{k+1,k}^{(\ell)}(R_1,R_2)= \left\langle t_{k+1,k}^{(\ell)}(U)\big(\rho_2(1+\rho_1^2/2+\rho_2^2/6)+O(W^{-2})\big)\right\rangle \notag\\&=
\left\langle e^{i(2k+1)\sigma}P_{k+1,k}^{(\ell)}(\cos\theta)\Big(iu_*W\cos\gamma\Tr[R_1,R_2-R_1]\tilde U(1+\rho_1^2/2+\rho_2^2/6)+O(W^{-2})\Big)\right\rangle \notag\\
&= c_{\ell,k}(\ell/W)[R_1,R_2-R_1]_{12}+O(\ell W^{-5/2}\log^3W))+O(\ell^2W^{-3}\log^2W).
\notag
\end{align}
Here $c_{\ell,k}$ is some bounded coefficient which appears from (\ref{as_Leg.0}) and after integration over $U$. 
We used also that because of integration over $\gamma,\delta,\sigma$
\begin{align*}
&\left\langle t_{k+1,k}^{(\ell)}(U)\rho_1^{\alpha}\rho_2^{\beta}\right\rangle=0,\quad (\alpha=1,3,\,\beta=0, 1,2)
\quad \left\langle t_{k+1,k}^{(\ell)}(U)(\rho_1^2+\rho_2^2)\right\rangle=O(\ell^2W^{-3}\log^2W).
\end{align*}
Here  we used that for any independent of $\tilde U$ matrix $B$  we have by (\ref{int}), (\ref{as_Leg.0})
\begin{align*}
& \left\langle t_{k+1,k}^{(\ell)}(\tilde U)(\Tr B\tilde U)^2\right\rangle=i(1+(k+1)/\ell)^{1/2}(1-k/\ell)^{1/2}\ell \,\left\langle\sin(\theta/2)e^{i(2k+1)\sigma+i\delta}
\right.\\ &\left.\times(1+O(\ell \sin^2(\theta/2)) 
 \Big((B_{11}-B_{22})\cos(\theta/2)\sin\sigma+\sin(\theta/2)\big(B_{21}e^{i\delta}+B_{12}e^{-i\delta}\big)\Big)^2
 \right\rangle\\
  =&\left\langle e^{i(2k+1)\sigma}\sin\sigma \sin^2(\theta/2)\cos (\theta/2)(1+O(\ell \sin^2(\theta/2))\right\rangle O(\ell\|B\|^2)
  +O(W^{-5}\|B\|^2)\\
  &=O(\ell^2 W^{-4}\|B\|^2).
\end{align*}
If $q=0$, then using (\ref{as_Leg.1})  and (\ref{F_ks.1}) one can easily check that
\begin{align*}
\mathcal{F}_{kk}^{(\ell)}(R_1,R_2)=\left\langle t_{k,k}^{(\ell)}(U)\right\rangle+O((\ell/W)^2 W^{-1}).
\end{align*}
It is easy to see that 
 $\left\langle t_{k,k}^{(\ell)}(U)\right\rangle\, (k=-\ell,\dots\ell)$ are eigenvalues of the operator 
 $\mathcal{E}^{(\ell)}K_{*R_1,R_2}\mathcal{E}^{(\ell)}$, where $K_{*R_1,R_2}$ is an integral operator with the kernel
$Z^{-1}(R_1,R_2)\exp\{k_{*R_1,R_2}(U_2^*U_1)\}$ (see (\ref{K,Z}) and (\ref{k_*})).
Hence, we need to compute  eigenvalues of $\mathcal{E}^{(\ell)}K_{*R_1,R_2}\mathcal{E}^{(\ell)}$. But
\begin{align*}
K_{*R_1,R_2}(U_2^*U_1)=K_{*R_1,R_2}(U_1U_2^*),
\end{align*}
so making the change of variables $U_2\to UU_1$ in the integral over $U_2$, we obtain
\begin{align*}
(K_{*R_1,R_2}t^{(\ell)}_{0k},t^{(\ell)}_{0p})=&\sum_s\int K_{*R_1,R_2}( U^*)t^{(\ell)}_{0s}(U)t^{(\ell)}_{sk}(U_1)\overline{t^{(\ell)}_{0p}(U_1)} dU dU_1\\
=&\delta_{kp}\int K_{*R_1,R_2}( U^*)t^{(\ell)}_{00}(U) dU.
\end{align*}
Thus, we get (\ref{Kt.0})  from (\ref{as_Leg.1}) and (\ref{sum_F}).

To prove (\ref{KPsi.1}) we  write
\[K_{R_1,R_2}h=\lambda_\ell h+r,\]
where  $r$  collects all the remainder terms (including $b^{(\ell)}_k$) from (\ref{Kt.0}). It is easy to check that the only remainder term which does not have a sufficient
bound for fixed $R_2$ is the one which contains $ [R_2-R_1,R_1]_{12}$. Let $\Psi(R)=p(R)e^{-\alpha u_*^2\Tr R^2}$ where $p(R)$ is a polynomials of 
degree at most $L$. Then we need to check that
\begin{equation}\label{com_term}
\Big\|\ell/W\int p(R_2)\mathcal{A}(R_1,R_2)e^{-\alpha u_*^2\Tr R_2^2}[R_1,R_2-\mu R_1]_{12}dR_2\Big\|\le CL \ell/W^2.
\end{equation}
Rewriting $\mathcal{A}$ in terms of $\mathcal{A}_0$, $\mathcal{A}_*$ (see (\ref{hat_A}) -- (\ref{A_*})), using (\ref{repr_A}),  
and integrating over $R_2$ by parts, we obtain
\begin{align}\label{com_parts}
&\int p(R_2)\mathcal{A}(R_1,R_2)e^{-\alpha u_*^2\Tr R_2^2}[R_1,R_2-\mu R_1]_{12}dR_2\\
&\qquad =(u_*^2W)^{-1}\int \mathcal{A}(R_1,R_2)e^{-\alpha u_*^2\Tr R_2^2}\partial\big(p(R_2),R_1\big)dR_2+O(\|\Psi\|W^{-3/2}),
\end{align}
where
$\partial\big(p(R_2),R_1\big)$ is some linear combination of entries  $R_1$ with the first derivatives of $p(R_2)$ with respect to $R_2$-entries.
%Here we also used $\|[R_1,R_2-\mu R_1]\|\sim W^{-1/2}$.
The additional multiplier  $W^{-1/2}$ in $O(\|\Psi\|W^{-3/2})$ appears because the derivatives of all additional terms other then $p(R_2)$ give the 
additional factor $W^{-c},\,c\ge 1/2$. 
 Now the $L_2$-norm of the last integral can be estimated as $O(L\|\Psi\|)$ in view of (\ref{der_Herm}),
which yields (\ref{com_term}) (notice that $W^{-3/2}\ll L/W$), thus (\ref{KPsi.1}).

Notice that if $p(R)=1$ then the integral in the r.h.s. of (\ref{com_parts})  is zero, and hence we obtain (\ref{KPsi0}).

To prove (\ref{diff_K_eps}), observe that it follows from the above arguments  that there are asymptotic expansions of coefficients of $K_{R_1,R_2}$ with
respect to $W^{-1}$. Since these expansions starts from $W^{-2}$ and the coefficients depend on $R_1,R_2$ via traces of some polynomials of
$R_1,R_2$
(except $(K_{R_1,R_2}t^{(\ell)}_{0k},t^{(\ell)}_{0k})$ which starts  from $1$, but $1$ does not depend on $R_1,R_2$), we conclude that (\ref{diff_K_eps})
is true.

To prove (\ref{b_K_R}), let us observe  that in view of  the bounds (\ref{b_rho})
\begin{align}\notag
&\|K_{*R_1,R_2}-K_{R_1,R_2}\|\le CW^{-1/2}\log^2 W\\
\Rightarrow& \mathcal{E}^{(\ell)}K_{R_1,R_2}\mathcal{E}^{(\ell)}\le
\mathcal{E}^{(\ell)}K_{*R_1,R_2}\mathcal{E}^{(\ell)}+CW^{-1/2}\log^2 W.
\label{diff_1}\end{align}
But in view of (\ref{as_Leg.1}) for $\ell>W^{3/4}\log^2W$, we have that
\[
\mathcal{E}^{(\ell)}K_{*R_1,R_2}\mathcal{E}^{(\ell)}\le 1-C'(W^{-1/2}\log^4W)^2+CW^{-1/2}\log^2 W\le  1-C'(W^{-1/2}\log^4W)^2/2.
\]

$\square$

Denote by $P_L$ the orthogonal projection in $\mathcal{H}_0=L_2(\mathcal{H}_{2,+})$  on the space $\mathcal{H}_L$
\begin{align}\notag
\mathcal{H}_L=&\mathrm{Lin}\{\Psi_{*\bar k}(R)\}_{|\bar k|\le L},\quad L=C_0\log^2W,  \\
\mathcal{P}_L=&P_L\otimes I\Big|_{L_2(U(2))}.
\label{P_L}\end{align}
We recall that  $\mathcal{K}_\zeta$ is an operator in $\mathcal{H}=\mathcal{H}_0\otimes L_2(U(2))$.
\begin{lemma}\label{l:2} For $L>C\log^2 W$ with sufficiently big $C$
\begin{align}\label{l2.1}
\|(I-\mathcal{P}_L)\mathcal{K}_\zeta(I-\mathcal{P}_L)\|\le (1-C_2L/W).
\end{align}
\end{lemma}
\textit{Proof}. 
Since 
\[
\|\mathcal{K}_\zeta-\mathcal{K}_0\|\le C\epsilon/W
\]
it suffices to prove (\ref{l2.1}) for $\mathcal{K}_0$. It is easy to see that $(I-\mathcal{P}_L)\mathcal{K}_0(I-\mathcal{P}_L)$ has a block-diagonal structure
with blocks
$(I-\mathcal{P}_L)\mathcal{E}_\ell\mathcal{K}_0\mathcal{E}_\ell(I-\mathcal{P}_L)$. By (\ref{Kt.0}) for $\Psi(U,R)\in (I-{P}_L)\mathcal{H}_0\otimes\mathcal{E}_\ell$
with $\ell<W^{3/4}\log W$ and $\|R_1-R_2\|\le W^{-1/2}\log W$ we have
\begin{align*}
&(K_{R_1,R_2}\Psi)(R,U)=\lambda_l\Psi(R,U)+O(\ell W^{-3/2}\log W)\\
\Rightarrow&(\mathcal{A}K_{R_1,R_2}\Psi,\Psi)=\lambda_l\int dU\int dR_1dR_2\mathcal{A}(R_1,R_2)\Psi(R_1,U)\Psi(R_2,U)+O(\ell W^{-3/2}\log W)\\
&\le (1-CL/2W)(1-C'\ell^2/W^2)+C''\ell W^{-3/2}\log W\le 1-2CL/4W.
\end{align*}
The last inequality here follows from
\[C'\ell^2/W^2-C''\ell W^{-3/2}\log W+C_2L/W>0\]
which is valid for all $\ell$ and any fixed $C',C'',C_2$, if we choose sufficiently big $C_0$ in (\ref{P_L}).
Here we used also that by (\ref{Delta})  for $\Psi\in (I-{P}_L)\mathcal{H}_0$
\[
(\mathcal{A}\Psi,\Psi)\le (\mathcal{A}_*\Psi,\Psi)+O(W^{-1})\le (1-CL/W)+O(W^{-1})\le 1-CL/2W.
\]
For $\ell\ge W^{3/4}\log W$ we use (\ref{b_K_R}) to write for $\Psi(R,U)\in\mathcal{H}_0\otimes\mathcal{E}^{(\ell)}$
\begin{align*}
(\mathcal{A}K_{R_1,R_2}\Psi,\Psi)\le& (1-CL/W)(\mathcal{A}\Psi_U,\Psi_U)\le  (1-CL/W)\|\Psi\|^2,\\
\Psi_U(R)=&\Big(\int dU|\Psi^2(R,U)|\Big)^{1/2}.
\end{align*}

%Let $u\in (I-\mathcal{P}_L)\mathcal H$, and $\|u\|=1$. According to the asymptotic of the Hermite polynomials with big $m$
%(see, e.g, Proposition 5.1.3 of \cite{PS:11}),  for $|x|\le c_{\alpha u_*^2}\sqrt m$
%\begin{align}\label{as_Herm}
%e^{-\alpha u_*^2x^2}H_{m}(u_*(2\alpha)^{1/2} x)/ \kappa_{m}=d_{\alpha u_*^2}m^{-1/4}
%\Big(\cos(m\phi_1(x/\sqrt m)+\phi_2(x/\sqrt m))+O(m^{-1})\Big),
%\end{align}
%where  $\phi_1(x),\phi_2(x)$ are some fixed smooth functions, and $c_{\alpha u_*^2}$, $d_{\alpha u_*^2}$-some fixed constants depending
% on $\alpha u_*^2$. Hence,  for the projector 
%  $\mathbb{P}_W$ on (\ref{Omega}) and for $u=(1-\mathcal{P}_L)u$ we have
%\[\|(1-\mathbb{P}_W)u\|^2>\frac{1}{2}\|u\|.
%\]
%Then we get by Lemma \ref{l:conc}
%\begin{align*}
%(\mathcal{K}_\zeta u,u)=&(\mathbb{P}_W \mathcal{K}_\zeta \mathbb{P}_W u,u)+((1-\mathbb{P}_W) \mathcal{K}_\zeta (1-\mathbb{P}_W)u,u)+O(W^{-3/2})\\
%\le &\|\mathbb{P}_Wu\|^2+(1-C_1\log W/W)\|(1-\mathbb{P}_W)u\|^2+O(W^{-3/2})\\ 
%=&(1-\|(1-\mathbb{P}_W)u\|^2)+(1-C_1\log W/W)\|(1-\mathbb{P}_W)u\|^2+O(W^{-3/2})\\=&1-C'\log W/W\|(1-\mathbb{P}_W)u\|^2+O(W^{-3/2})\le 
%1-C_2\log W/W.
%\end{align*}
$\square$

\medskip

Recall that $\mathcal{K}_0=\mathcal{K}_\zeta\Big|_{\zeta=0}$ and set
\begin{align}\label{bb_K_0}
\mathbb{K}_0=\mathcal{E}_M\mathcal{K}_0\mathcal{E}_M\,\quad M=\max\{C \log W, C_0(\epsilon W)^{1/2}\}.
\end{align}
The following lemma gives an information about the eigenvalues and eigenvectors of $\mathbb{K}_0$:
\begin{lemma}\label{l:la_max} 
For any $\ell\le M,$ $\mathbb{K}_0$ has $2\ell+1$ eigenvalues $\lambda_{\ell,k}$ with eigenvectors $\Psi_{\ell,k}(R,U)$  such that
\begin{align}
%\label{de_la}
%&\lambda_{\ell,k}\le \lambda_{\max}\lambda_\ell+O(W^{-5/2}\log^3W),\\
&|\lambda_{\ell,k}-\lambda_{\max}|\le C(\ell/W)^2,
\label{de_la.0} \end{align}
where $\lambda_\ell$ is defined in (\ref{lam_l})

Moreover, for any fixed $p>0$,
there are vectors $h_{\bar j,\ell,k}\in \mathcal{E}^{(\ell)}$ such that $\|h_{\bar j,\ell,k}\|\le C$ and
\begin{align}\label{de_Psi}
&\Psi_{\ell,k}(R,U)=\Psi_{*\bar 0}(R)h_{\bar 0,\ell,k}(U)+\sum_{s=1}^{2p-1}\sum_{|\bar j|\le s+2}W^{-s/2}\Psi_{*\bar j}(R)h_{\bar j,\ell,k}(U)+O(W^{-p})
\end{align}
with $\Psi_{*\bar j}$ defined in (\ref{psi_k}).
\end{lemma}

\textit{Proof.} 
%Observe  that  (\ref{de_la}) follows from (\ref{KPsi0}) and (\ref{de_Psi}).
% Indeed,  represent
% \[ 
% \Psi_{\ell,k}=\Psi_{\bar 0}(R)h_{0,\ell,k}(U)+\tilde\Psi_{\ell,k}^{(2)}(R,U),\quad \quad\int dR \Psi_{\bar 0}(R)\tilde\Psi_{\ell,k}^{(2)}(R,U)=0.
% \]
%Then we obtain in view of (\ref{KPsi0})
%\begin{align*}
%(\mathbb{K}\Psi_{\ell,k},\Psi_{\ell,k})=&(\mathbb{K}_0\Psi_{\bar 0}\otimes h_{0,\ell,k},\Psi_{\bar 0}\otimes h_{0,\ell,k})
%+2(\mathbb{K}_0\Psi_{\bar 0}\otimes h_{0,\ell,k},\Psi_{\ell,k}^{(2)})
%+(\mathbb{K}_0\Psi_{\ell,k}^{(2)},\Psi_{\ell,k}^{(2)})\\
%\le &\lambda_{\max}\lambda_\ell\|h_{0,\ell,k}\|^2+O(W^{-5/2}\log^2W)+\lambda_{\max}(1-C/W)\|\Psi^{(2)}\|^2\\
%\le& \lambda_{\max}\max_{x^2+y^2=1}\{\lambda_\ell x^2+(1-C/W)y^2\}+O(W^{-5/2}\log^2W)\\
%\le& \lambda_{\max}\lambda_\ell+O(W^{-5/2}\log^2W).
%\end{align*}
%Here the bounds for the first and the second terms in the r.h.s.  follow from  (\ref{KPsi0}), and the bound for the third term is an analogue
%of the third inequality of (\ref{de_la.1}) below. The proof of (\ref{de_la.1}) is  given at the end of the proof of the lemma.

To prove (\ref{de_la.0}) we consider $\mathcal{P}_{\bar j}$ - the orthogonal projection on $\Psi_{*\bar j}$, set
%where $\{\Psi_{\bar j}\}$ are eigenvectors of $\mathcal{A}$ defined in Lemma \ref{l:A}. Set
\begin{align*}
\mathbb{K}_{0,(\bar j,\bar j')}=\mathcal{P}_{\bar j}\mathbb{K}_{0}\mathcal{P}_{\bar j'},
\end{align*}
and consider $\mathbb{K}_0$ as a $2\times 2$ block matrix, with 
\begin{align*}
&\mathbb{K}_0^{(11)}=\mathcal{P}_{\bar 0}\mathbb{K}_{0}\mathcal{P}_{\bar 0},\quad 
\mathbb{K}_0^{(22)}=(1-\mathcal{P}_{\bar 0})\mathbb{K}_{0}(1-\mathcal{P}_{\bar 0}), \\
&\mathbb{K}_0^{(12)}=\mathcal{P}_{\bar 0}\mathbb{K}_{0}(1-\mathcal{P}_{\bar 0}),\quad \mathbb{K}_0^{(21)}
=(1-\mathcal{P}_{\bar 0})\mathbb{K}_{0}\mathcal{P}_{\bar 0}.
\end{align*}
Then  (\ref{de_la.0}) follows from the bound:
\begin{align}\label{de_la.1}
&\|\mathbb{K}_0^{(11)}-\lambda_{\max}\|\le C(\ell/W)^2\quad \|\mathbb{K}_0^{(12)}\|\le CW^{-3/2},\\
&\mathbb{K}_0^{(22)}\le \lambda_{\max}-C/W.
\notag\end{align}
Indeed,   the last two inequalities of (\ref{de_la.1}) imply that for $|z-\lambda_{\max}|\le C(\ell/W)^2$ 
\begin{align*}
\mathbb{X}=\mathbb{K}_0^{(11)}-z-\mathbb{K}_0^{(12)}(\mathbb{K}_0^{(22)}-z)^{-1}\mathbb{K}_0^{(12)}=\mathbb{K}_0^{(11)}-z+O(W^{-2}).
\end{align*}
Hence, for  $|z-\lambda_{\max}|\le C(\ell/W)^2$ all eigenvalues of $\mathbb{X}$ differ from corresponding eigenvalues of $\mathbb{K}_0^{(11)}-z$
less than $cW^{-2}$. Then the first inequality of (\ref{de_la.1}) gives us (\ref{de_la.0}). In addition,  since for
$C(\ell/W)^2\le |z-\lambda_{\max}|\le c/W$ with some big enough $C$, one can conclude that $\mathbb{X}^{-1}$
 exists for such $z$. Hence,  there are no eigenvalues of $\mathbb{K}_0$ in the domain
$C(\ell/W)^2\le |z-\lambda_{\max}|\le c/W$. Thus, to finish the proof of (\ref{de_la.0}), it is sufficient to prove (\ref{de_la.1}).
The first inequality follows from Lemmas \ref{l:Z_0}, \ref{l:KPsi}. The second inequality follows from (\ref{exp_dec}) below.
The proof of third bound of  (\ref{de_la.1}) is given at the end of the proof of Lemma \ref{l:la_max}.
% The third bound follows from the inequality (see (\ref{KPsi.1})):
%\[ \|\mathbb{K}_0-\mathcal{A}\otimes\mathcal{E}_M\|\le C(M/W)^2\ll W^{-1}
%\]
%This inequality implies  that all eigenvalues of $\mathbb{K}_0$  differ from eigenvalues of $\mathcal{A}\otimes\mathcal{E}_M$
%not more than by $(M/W)^2$. Since $\mathcal{A}\otimes\mathcal{E}_M$ has $(M+1)^2 $ first eigenvalues equal to $\lambda_{\max}$
%and all other eigenvalues less than $\lambda_{\max}-C/W$, we obtain the third bound of (\ref{de_la.1}).

\medskip

Let us prove (\ref{de_Psi}). Consider the eigenvector $\Psi_{\ell,k}(R,U)$ of $\mathbb{K}_0$ corresponding to $|\lambda_{\ell,k}-\lambda_{\max}|\le C\log^2W/W^2$.
Let 
\[\Psi_{\ell,k}(R,U)=(\Psi_{\ell,k}^{(1)},\Psi_{\ell,k}^{(2)})=(\mathcal{P}_{\bar 0}\Psi_{\ell,k},(I-\mathcal{P}_{\bar 0})\Psi_{\ell,k})
\]
be decomposition of $\Psi_{\ell,k}$.
 Since $(\Psi_{\ell,k}^{(1)},\Psi_{\ell,k}^{(2)})$ is an eigenvector of $\mathbb{K}_0$,
it satisfies the equation
\begin{align*}
%&\mathcal{A}^{(11)}-z-\sum\mathcal{A}^{(12)}_{1i}(\mathcal{A}^{(22)}-z)^{-1}_{ij}\mathcal{A}^{(21)}_{j1}=0,\\
&\mathbb{K}_0^{(12)}\Psi_{\ell,k}^{(1)}+(\mathbb{K}_0^{(22)}-\lambda_{\ell,k})\Psi_{\ell,k}^{(2)}=0.
\end{align*}
thus
\begin{align}
\Psi_{\ell,k}^{(2)}=-(\mathbb{K}_0^{(22)}-\lambda_{\ell,k})^{-1}\mathbb{K}_0^{(12)}\Psi_{\ell,k}^{(1)},
%\quad\|\mathcal{A}^{(12)}\|\le \epsilon C_0W^{-1}, \quad \|\mathcal{A}^{(21)}_{M_1,M_2}\|\le \epsilon C_0W^{-1}.
\label{Psi^2}\end{align}
Given that the third inequality of (\ref{de_la.1}) is valid,  we have 
\begin{align}\label{1_Psi_2}
&\|(\mathbb{K}_0^{(22)}-\lambda_{\ell,k})^{-1}\|\le CW.
%& \Psi_0^{(2)}=\Psi_0^{(1)}\mathbb{K}_0^{(21)}(\mathbb{K}_0^{(22)}-\lambda_{\ell,k})^{-1},
%%\quad\|\mathcal{A}^{(12)}\|\le \epsilon C_0W^{-1}, \quad \|\mathcal{A}^{(21)}_{M_1,M_2}\|\le \epsilon C_0W^{-1}.
%\label{Psi^2}
\end{align}
Assume that for any $p$ we prove the bound
\begin{align}\label{exp_dec}
\|\mathbb{K}_{0(\bar j,\bar k)}\|\le \tilde C_p\big(\min\{W^{-3/2},W^{-|\bar j-\bar k|/2}\}+W^{-p-1}\big),\quad \bar j\not=\bar k,\quad \min\{|\bar j|,|\bar k|\}\le L.
\end{align}
Introduce the matrix $\tilde{\mathbb{K}}_0$ which is obtained from $\mathbb{K}_0^{(22)}$ if we replace all entries 
$\mathbb{K}_{0,(\bar j,\bar j')}$ with $|\bar j-\bar j'|\ge 2p+2\wedge\min\{|\bar j|,|\bar j'|\}\le 2L$ by zeros.
It is easy to see that
\[
\|(\mathbb{K}_{0}^{(22)}-\lambda_{\ell,k})^{-1}-(\tilde{\mathbb{K}}_{0}-\lambda_{\ell,k})^{-1}\|\le CW^{-p}.
\] 
Consider $\tilde{\mathbb{K}}_{0}$ as a block matrix such that
\[
\tilde{\mathbb{K}}_{0}^{(11)}=\Big(\sum_{ 1\le|\bar m|<L}\mathcal{P}_{\bar m}\Big)\tilde{\mathbb{K}}_{0}\Big(\sum_{1\le|\bar m|<L}\mathcal{P}_{\bar m}\Big),\quad
\tilde{\mathbb{K}}_{0}^{(22)}=\Big(1-\sum_{ 1\le|\bar m|<L}\mathcal{P}_{\bar m}\Big)\tilde{\mathbb{K}}_{0}\Big(1-\sum_{1\le|\bar m|<L}\mathcal{P}_{\bar m}\Big).
\]
Observe that $\tilde{\mathbb{K}}_{0}^{(11)}$ contains only a finite number of diagonals and $\tilde{\mathbb{K}}_{0(\bar j,\bar j')}^{(12)}\not=0$ only if 
$|L-|\bar j||< 2p+2$ and $|\bar j'|< 2p+2$, and so $\tilde{\mathbb{K}}_{0(\bar j,\bar j')}^{(12)}$ contains only finitely many (depending on $p$) nonzero entries.
Hence, denoting $\hat{\mathbb{K}}_{0}=\mathrm{diag}\{\tilde{\mathbb{K}}_{0}^{(11)},\tilde{\mathbb{K}}_{0}^{(22)}\}$,
we get for any fixed $|\bar j_0|,|\bar j_0'|<L/3$
\begin{align}\notag
&\|(\tilde{\mathbb{K}}_{0}-\lambda_{\ell,k})^{-1}_{(\bar j_0,\bar j_0')}-(\hat{\mathbb{K}}_{0}-\lambda_{\ell,k})^{-1}_{(\bar j_0,\bar j_0')}\|
\le \sum_{|L-|\bar j||< 2p+2} \|(\hat{\mathbb{K}}_{0}^{(11)}-\lambda_{\ell,k})^{-1}_{(\bar j_0,\bar j)}\tilde{\mathbb{K}}_{0(\bar j,\bar j')}^{(12)}
(\tilde{\mathbb{K}}_{0}-\lambda_{\ell,k})^{-1}_{(\bar j',\bar j_0')}\|\\
&\le C_pW\max_{|L-|\bar j||< 2p+2}\| (\hat{\mathbb{K}}_{0}^{(11)}-\lambda_{\ell,k})^{-1}_{(\bar j_0,\bar j)}\|=C_pW\max_{|L-|\bar j||< 2p+2}\| (\tilde{\mathbb{K}}_{0}^{(11)}-\lambda_{\ell,k})^{-1}_{(\bar j_0,\bar j)}\|.
\label{diff_bbK}\end{align}
Here we  used (\ref{1_Psi_2}). But  if we consider $\tilde{\mathbb{K}}_{0}^{(11)}-\lambda_{\ell,k}$ as a sum of its diagonal $\mathbb{K}_{d}$
and off diagonal $\mathbb{K}_{off}$ parts then
\[(\tilde{\mathbb{K}}_{0}^{(11)}-\lambda_{\ell,k})^{-1}=\sum \mathbb{K}_{d}^{-1}(\mathbb{K}_{off}\mathbb{K}_{d}^{-1})^s,
\]
one can see easily that, in view of  (\ref{exp_dec}),
\[
\|(\tilde{\mathbb{K}}_{0}^{(11)}-\lambda_{\ell,k})^{-1}_{(\bar j_0,\bar j)}\|\le W (CW^{-1/2})^{|\bar j_0-\bar j|}.
\]
Since we consider $|j_0|<L/3$, we have$|\bar j_0-\bar j|>L/2$  in the last line of (\ref{diff_bbK}),  and so
\[
\|(\tilde{\mathbb{K}}^{(22)}-\lambda_{\ell,k})^{-1}_{(\bar j_0,\bar j)}\|\le W (CW^{-1/2})^{|\bar j_0-\bar j|}+CW^{-p}.
\]
Now  (\ref{Psi^2}) and (\ref{exp_dec}) imply (\ref{de_Psi}).

\medskip

To finish the proof, we are left to check (\ref{exp_dec}). 
Repeating the argument of Lemma \ref{l:Z_0}, we conclude that to find $\mathbb{K}_{0(\bar j,\bar k)}$ one should compute the
sum (with some $W$-independent coefficient) of the integrals
\begin{align*}
I_{\bar j,\bar k}(m,\tilde{p}_\ell,)=\int dR_1dR_2 &\mathcal{A}_*(R_1,R_2)e^{-u_*^2\alpha\Tr R_2^2}P_{\bar j}(R_2)P_{\bar k}(R_1)
\\
&\times p_{\ell}(R_1/W^{1/2},R_2/W^{1/2})\prod_{s=1}^{2m} [R_1,R_2-\mu R_1]_{\alpha_s\beta_s}.
\end{align*}
Hear $\mathcal{A}_*(R_1,R_2)$ and $\mu$ are written as in (\ref{repr_A}), $P_{\bar j}(R_2),P_{\bar k}(R_1)$ are the products of the Hermite polynomials   (see (\ref{psi_k})), and $p_{\ell}(R_1/W^{1/2},R_2/W^{1/2})$ is some uniform polynomials
of degree $\ell$ of $R_1,R_2$ written as in (\ref{expR}). Integrating by parts $2m$ times with respect to $R_2$ and using the recurrent formulas for the Hermite polynomial 
and their derivatives, we conclude that
\begin{align*}
I_{\bar j,\bar k}(m,\tilde{p}_\ell)=&O(W^{-(2m+\ell+1)/2}),\quad |\bar j-\bar k|>2m+\ell,\\
I_{\bar j,\bar k}(m,\tilde{p}_\ell)\le &CW^{-(\ell+2m)/2},\quad |\bar j-\bar k|\le 2m+\ell.
\end{align*}
These relations prove (\ref{exp_dec}).

\medskip

The proof of the last bound of (\ref{de_la.1}) is based on the simple proposition
\begin{proposition}\label{p:norm} Given a $2\times 2$  block matrix $\mathbb{M}=\mathbb{M}^*$ with blocks $\mathbb{M}^{(\alpha\beta)}$, such that
\[ \mathbb{M}^{(11)}<m_1,\quad \mathbb{M}^{(22)}<m_2<m_1.\]
Then 
\begin{align}\label{p:n.1}
\lambda_{\max}(\mathbb{M})\le \lambda_*= m_1+\|\mathbb{M}^{(12)}\|^2|m_2-m_1|^{-1}.
\end{align}
\end{proposition}
\textit{Proof of Proposition \ref{p:norm}}. Bound (\ref{p:n.1}) follows from the inequality valid for any $\lambda> \lambda_*$ :
\[
\mathbb{M}^{(11)}-\lambda-\mathbb{M}^{(12)}(\mathbb{M}^{(22)}-\lambda)^{-1}\mathbb{M}^{(21)}\le m_1-\lambda
+\|\mathbb{M}^{(12)}\|^2|m_2-m_1|^{-1}=\lambda_*-\lambda
\]
Hence, the matrix in the l.h.s. is invertible, and since $M^{(22)}-\lambda$ is also invertible, we conclude  that $M-\lambda$ is invertible for 
$\lambda> \lambda_*$.

$\square$

\textit{Proof of the last bound of (\ref{de_la.1}).} Consider 
\[ \mathbb{M}=\mathbb{K}_0^{(22)},\quad \mathbb{M}^{(11)}=P_L\mathbb{K}_0^{(22)}P_L,\quad \mathbb{M}^{(22)}=(I-P_L)\mathbb{K}_0^{(22)}(I-P_L),\quad
\mathbb{M}^{(12)}=P_L\mathbb{K}_0^{(22)}(I-P_L).
\]
Then (\ref{KPsi.1}) yields
\[
 \mathbb{M}^{(11)}\le \lambda_{\max}-c/W,\quad \|M^{(12)}\|\le CW^{-3/2},
\]
and (\ref{l2.1}) implies
\[ \mathbb{M}^{(22)}\le 1-C\log W/W.\]
Choosing $\delta=c/2W$ we obtain the last bound of (\ref{de_la.1}).

$\square$

In the following lemma we study the action of $\mathcal{K}_\zeta$ on the vectors from  $\mathcal{H}_L\otimes \mathcal{E}_M$. 
An important role below belongs to the vectors of the form
\begin{align}\label{Psi_*}
 \Psi_{\epsilon,h}(R,U)=\Psi(R-\epsilon \mathcal{M}(U))h(U),\quad\Psi\in\mathcal{H}_L,\quad h(U)\in \mathcal{E}_{2M},
\end{align}
 with  $\mathcal{M}(U)$ of (\ref{M(U)}) and $\epsilon=(W/N)^{1/2}$.

In what follows
it will be convenient to apply $\mathcal{K}_\zeta$ to the vectors constructed from eigenvectors of $\mathcal{K}_0$ or $\mathcal{A}$ of (\ref{A}).
But to apply Lemma \ref{l:Z_0} or Lemma \ref{l:KPsi} to some vector $\Psi(R,U)$, we need to know that $\Psi(R,U)$ can be expanded in a sum
of vectors belonging to $\mathcal{H}_L$. Hence in the different places below we are using the following simple observation.
Since by the condition of Theorems \ref{t:Gin} -- \ref{t:locGin} we have
$W>N^{\varepsilon_0}$, one can choose some $W,N$-independent $p$ such that $W^{p}>N^4$. If
$\Psi_{\ell,k}$ is an eigenvector of  $\mathbb{K}_0$  of (\ref{bb_K_0}) with eigenvalue $\lambda_{\ell,k}$ 
satisfying  (\ref{de_la.0}), then taking this  $p$ in (\ref{de_Psi}) sufficiently big and denoting $\tilde \Psi_{\ell,k}$ the r.h.s. of (\ref{de_Psi}) without the
 remainder $O(W^{-p})$, we have
\begin{align}\label{Psi^p}
\Psi_{k,\ell}=\tilde \Psi_{\ell,k}+O(N^{-2}),\quad \mathbb{K}_0\tilde \Psi_{\ell,k}=\lambda_{\ell,k}\tilde \Psi_{\ell,k}+O(N^{-2}).
\end{align}
Thus, applying any assertion of Lemmas \ref{l:Z_0}, \ref{l:KPsi} to $\Psi_{k,\ell}$, we replace it by $\tilde \Psi_{\ell,k}$,
 then apply the assertion which we need, and then come back to  $\Psi_{k,\ell}$, using that the error of the replacement is very small.
 
 The same argument allows us to apply  assertions of Lemmas \ref{l:Z_0}, \ref{l:KPsi} to vectors $\Psi^{(\mu)}_{|\bar j|}$ described in
 Lemma \ref{l:A}.  Using (\ref{diff_eig}) and $\tilde\Psi^{(\mu)}_{|\bar j|}$ (which are analogues of $\tilde \Psi_{\ell,k}$) belong to $\mathcal{H}_{L+p}$,
we conclude that assertions of  Lemmas \ref{l:Z_0}, \ref{l:KPsi}  are valid for them.

\begin{lemma}\label{c:1} 
Given any function of the form (\ref{Psi_*})
we have
\begin{align}\label{cor:1.0}
(\mathcal{K}_\zeta \Psi_{\epsilon,h})(R_1,U_1)= e^{2\nu(U_1)/N}(\mathcal{K}_0\Psi_{0,h})\big(R_1-\epsilon\mathcal{M}(U_1)\big)
+O(\epsilon W^{-3/2}+\epsilon L M/W^2).
\end{align}
where  $\nu$ is defined in (\ref{nu}).
For  functions of the form
\begin{align}\label{Psi_eps}
\Psi_{\ell,k,\epsilon}(R,U)=\Psi_{\ell,k}(R-\epsilon\mathcal{M}(U),U)
\end{align}
with $\{\Psi_{\ell,k}(R)\}$ defined in Lemma \ref{l:la_max}, we have
\begin{align}\label{cor:1.2}
(\mathcal{K}_\zeta \Psi_{\ell,k,\epsilon},\Psi_{\ell',k',\epsilon})=&\delta_{\ell,\ell'}\delta_{k,k'}\lambda_{\ell}+O(N^{-1}
+\epsilon^2 W^{-3/2}+\epsilon(\ell/ W)^{2}),\quad \max\{\ell,\ell'\}\ge 1,\\
(\mathcal{K}_\zeta\Psi_{0,0,\epsilon},\Psi_{0,0,\epsilon})
 =&\lambda_{\max}+O(\epsilon N^{-1}+\epsilon^2W^{-3/2}),
\label{cor:1}\end{align}
with $\lambda_\ell$ of (\ref{lam_l}).
 \end{lemma}

 \textit{Proof}.
 Expand $F_\zeta(R_2,U_2)\tilde \Psi(R_2-\epsilon\mathcal{M}(U_2))$  into a series with respect to $\epsilon$.
 Note that if $U$ is written as in (\ref{U}), then
  \[\Tr \phi (R)\mathcal{M}(U)=a(R)\cos\theta+\sin\theta (b(R)e^{i\psi}+\bar b(R)e^{-i\psi})
  \]
  Hence, each term of the expansion with respect to $\epsilon$ can be written  in terms of operators $\hat\Phi_1$ 
  and $\hat\Phi_2$ of multiplication by $\cos\theta$  and $\sin\theta $.
 We use the representation $t^{(\ell)}_{0k}$ in terms of the associated Legendre polynomials (see (\ref{as_P})),
and   the recursion formulas  
\begin{align}\label{rec_Leg}
\cos\theta \,P^{(\ell)}_{0k}(\cos\theta)=&c_{\ell,k}P^{(\ell+1)}_{0k}(\cos\theta)+d_{\ell,k}P^{(\ell-1)}_{0k}(\cos\theta),\\
\sin\theta \,P^{(\ell)}_{0k}(\cos\theta)=&c_\ell\big(P^{(\ell+1)}_{0k+1}(\cos\theta)-P^{(\ell-1)}_{0k+1}(\cos\theta)\big).
\notag\end{align}
Here $c_{\ell,k},d_{\ell,k},c_\ell$ are some bounded uniformly in $k,\ell$
coefficients, whose concrete form of is not important for us. 

 Then, by (\ref{KPsi0}) we have for any $h\in\mathcal{E}_{M}$
 \begin{align*}
 \int dR_2  B(R_1-R_2) \mathcal{Z}(R_1,R_2)F_0(R_2)\Psi_{0}(R_2)([\Phi_{\alpha},K_{R_1,R_2}]h)=O(W^{-1/2}(\ell/W)^{2}),\quad \alpha=1,2,
 \end{align*}
where $[.,.]$ denotes a  commutator. Hence, for operator of multiplication by $F_\zeta(R,U)$ the error term for the commutator
is  $O(\epsilon^s W^{-1/2}(\ell/W)^{2})$.
Notice that zero order with respect to $\epsilon$ term contain $e^{\nu(U_2)/N}$, and the commutator with this term gives us an error $O(N^{-1} W^{-1/2}(\ell/W)^{2})$. Therefore,
\begin{align}
\label{cor:1.3}
(\mathcal{K}_\zeta \Psi_{\epsilon,h})(R_1,U_1)=& F_\zeta(R_1,U_1)
\int B(R_1-R_2) \mathcal{Z}(R_1,R_2)F_\zeta(R_2,U_1)\\ 
&\times 
\Psi(R_2-\epsilon \mathcal{M}(U_1))(K_{R_1,R_2}h)(U_1)dR_2+O(\epsilon W^{-1/2}(\ell/W)^{2}).
\notag\end{align}
Then we replace $F_\zeta(R_1,U_1)$ by $F_0(R_1-\epsilon\mathcal{M}(U_1))$ with an error $O(\epsilon W^{-3/2})$, 
using that in view of
(\ref{pert}) and   (\ref{ti-f})
\begin{align}\label{f_1}
&F_\zeta(R,U)=F_0(R-\epsilon\mathcal{M}(U))e^{f_1(R,U)},\\
&f_1(R,U)=C_1\nu(U)/N+C_2\epsilon W^{-3/2}\Tr\mathcal{M}(U)R^2\varphi_2(1+R/W^{1/2}),
\notag\end{align}
where $\varphi_2(R)$ is some analytic function obtained from $\varphi_0(R)$ and $\varphi_1(R)$ of (\ref{ti-f}).

Finally, using (\ref{diff_Z}) and (\ref{diff_K_eps}),  we replace
$\mathcal{Z}(R_1,R_2)$ by $\mathcal{Z}(R_1-\epsilon\mathcal{M}(U_1),R_2-\epsilon\mathcal{M}(U_1))$ with an error $O(\epsilon/W^2)$, 
and $K_{R_1,R_2}$ by $K_{R_1-\epsilon\mathcal{M}(U_1),R_2-\epsilon\mathcal{M}(U_1)}$ with an error $O(\epsilon\log^2W/W^2)$.
Thus, integrating over $R_2$ and changing $R_2-\epsilon\mathcal{M}(U_1)\to R_2$, we get (\ref{cor:1.0}).

It follows directly from (\ref{cor:1.0}), that
\begin{align}
(\mathcal{K}_\zeta \Psi_{\ell,k,\epsilon})(R,U)=&\lambda_{\ell}e^{2\nu(U_1)/N}\Psi_{\ell,k,\epsilon}(R,U)+O(\epsilon W^{-3/2}).
\end{align}
The term $O(\epsilon L\ell W^{-2})$ becomes  $O(\epsilon \ell W^{-5/2})$ by (\ref{de_Psi}). 
Thus, we need only to check that if we take the scalar product of the l.h.s. with $\Psi_{\ell',k',\epsilon}$, then
the  term of order $O(\epsilon W^{-3/2})$ disappears. We recall that the term appears because of the replacement of $F_\zeta(R,U)$
by $F_0(R-\epsilon\mathcal{M}(U)$ (see (\ref{f_1})). Therefore, its contribution to the scalar product will have the form
\begin{align*}
&\epsilon W^{-3/2}\int\Tr\mathcal{M}(U)R^2\varphi_2(1+R/W^{1/2})\Psi_{\ell,k}(R,U)\Psi_{\ell', k'}(R,U)dRdU\\=
&\epsilon W^{-3/2}\int\Tr\mathcal{M}(U)R^2\varphi_2(1+R/W^{1/2})\Psi^2_{0,0}(R)dRh^{(\ell)}_{k}(U)h^{(\ell')}_{k'}(U)dU+O(\epsilon W^{-2}),
\end{align*}
where  we used (\ref{de_Psi}) to replace $\Psi_{\ell, k}(R,U)$ by $\Psi_{0,0}(R)h^{(\ell)}_{k}(U)+O(W^{-1/2})$ and $\Psi_{\ell', k'}(R,U)$ by 
$\Psi_{00}(R)h^{(\ell')}_{k'}(U)+O(W^{-1/2})$.

In order to compute the last integral, we observe that $\Psi_{0,0}$ is invariant with respect to the 
change $R\to VRV^*$ with any unitary $V$.
Making this change and integrating with respect to $dV$, we obtain for any $\tilde \varphi$ and any matrix
 $\mathcal{M}:\mathcal{M}=\mathcal{M}^*,\,\Tr \mathcal{M}=0$
\begin{align}\label{int_2}
\int dR\Psi_{00}^2(R_1)\Tr (\tilde \varphi(R)\mathcal{M})=\int \Psi_{00}^2(R)\Tr (V\mathcal{M}V^*\tilde \varphi(R))dVdR=0,
\end{align}
since
\[
\int (V\mathcal{M}V^*)_{\alpha,\beta}dV=0.
\]

To prove (\ref{cor:1}) we need to check that for $\ell=0,k=0$ the linear with respect to $\epsilon$ error terms  in (\ref{cor:1.0}) disappear. 
Let us check that for any $U$ if we set
\begin{align}\label{I(eps)}
\mathcal{A}_{\epsilon,U}(R_1,R_2)=&B(R_1-R_2)\mathcal{Z}(R_1,R_2)F_0(R_1-\epsilon \mathcal{M}(U))F_0(R_2-\epsilon \mathcal{M}(U)),
\end{align}
then
\begin{align}\label{I(eps).1} 
I(\epsilon )=&(\mathcal{A}_{\epsilon,U}\Psi_{0,0,\epsilon},\Psi_{0,0,\epsilon})-(\mathcal{A}\Psi_{0,0},\Psi_{0,0})=O(\epsilon^2W^{-2}),
\end{align}
Since $I(\epsilon)$ could be written in the form
\begin{align*}
I(\epsilon )=&\int B(R_1-R_2)\Big(\mathcal{Z}(R_1,R_2)-\mathcal{Z}(R_1-\epsilon \mathcal{M}(U),R_2-\epsilon \mathcal{M}(U))\Big)\\
&\times F_0(R_1-\epsilon \mathcal{M}(U))F_0(R_2-\epsilon \mathcal{M}(U)))
\Psi_{0,0}(R_1-\epsilon \mathcal{M}(U))\Psi_{0,0}(R_2-\epsilon \mathcal{M}(U))dR_1dR_2,
\end{align*}
(\ref{diff_Z})  implies that
$|I(\epsilon)|\le C\epsilon W^{-2}$.  On the other hand, $I(\epsilon )$ for any $\epsilon$ can be expand in the asymptotic series 
with respect to $W^{-1/2}$,  i.e. for any $p>0$
\[
I(\epsilon)=\sum_{k= 4}^pW^{-k/2}\psi_k(\epsilon)+O(W^{-(p+1)/2}),
\]
where $\{\psi_k\}$ are analytic in epsilon functions. Hence, it is sufficient to check that $I'(0)=0$. But since $\Psi_{0,0}$ is an eigenvector
of $\mathcal{A}$ corresponding to $\lambda_{\max}$, we get
\begin{align*}
I'(0)=-\lambda_{\max}\int\Psi_{0,0}^2(R)\Tr \mathcal{M}(U)\Big(R(u_*^2+\alpha/W)+W^{-3/2}\phi(R)\Big) dR.
\end{align*}
Using (\ref{f_1}), we get
\begin{align}\label{cor:1.6}
&(\mathcal{A}_{\epsilon,U}(e^{f_1}\Psi_{0,\epsilon}),e^{f_1}\Psi_{0,\epsilon})-e^{2\nu/N}(\mathcal{A}_{\epsilon,U}\Psi_{0,\epsilon},\Psi_{0,\epsilon})\\
=&\epsilon W^{-3/2}\lambda_{\max}e^{2\nu/N}\int\Psi_0^2(R)\Tr (\mathcal{M}(U)\varphi_2(R))dR+O(\epsilon^2W^{-3/2})=O(\epsilon^2W^{-3/2}).
\notag\end{align}
%Combining this with (\ref{I(eps).1}), we obtain (\ref{cor:1.2}). 
%
%To prove (\ref{cor:1}), we just observe that
%\[(( K_{R_1,R_2}-I)e^{\nu/N}(F_\zeta\Psi_{0,0,\epsilon}),1)=0,
%\]
%hence the term which gave us  $O(\epsilon W^{-2})$ in (\ref{cor:1.5}) disappears. 
In addition,
 \[\int \nu(U)dU=0\Rightarrow \int e^{2\nu(U)/N}dU=O(N^{-2}).
 \]
 Combining this with (\ref{I(eps).1}) and (\ref{cor:1.6}), we obtain (\ref{cor:1}).
$\square$

\section{Proofs of Theorems \ref{t:Gin}, \ref{t:locGin}.}\label{s:Proof}

\begin{lemma}\label{l:denom}
 Given $\tilde\Theta(z_1,z_2)$ of the form (\ref{ti-theta}), and $N>CW\log W$ with sufficiently big $C$, we  have
\begin{align}\label{cor:2.1}
&\lim_{N\to\infty,\frac{W^2\log N}{N}\to 0}\tilde\Theta(z,z)=\lambda_{\max}^{N-1}g_1^2(1+o(1)),\quad g_1=(g,\Psi_{*\bar 0}),\quad W>N^{\varepsilon_0},\\
& \lim_{N\to\infty,\frac{W^2}{N\log N}\to \infty}\tilde\Theta^{1/2}(z+\zeta/N^{1/2},z+\zeta/N^{1/2})
%\tilde\Theta^{1/2}(z-\zeta/N^{1/2},z-\zeta/N^{1/2})
=e^{2|\zeta|^2}(1+o(1)).
\notag\end{align}
\end{lemma}

\textit{Proof.}
Observe  that for any $z'$  $\tilde\Theta(z',z')$ does not contain integration with respect to the unitary group. Moreover, by (\ref{diff_eig}) and (\ref{psi_k})
the spectral gap of $\mathcal{A}$ between $\lambda_{\max}$ and the next eigenvalue is bigger than $c/W\gg N^{-1}$. In particular, for $\zeta=0$
we have
\begin{align*}
\tilde\Theta(z,z)=(\mathcal{A}^{N-1}g,g)=\lambda_{\max}^{N-1}(g,\Psi_{00})^2+O(e^{-Nc/W}\|g\|^2)=\lambda_{\max}^{N-1}g_1^2(1+o(1)).
\end{align*}
since by (\ref{de_Psi.0}) $(g,\Psi_{00})=(g,\Psi_{*\bar 0})+o(1)$.

For $z'=z+\zeta/N^{1/2}$, replacing $L$ by $\pm I$ in (\ref{pert}), we get
\begin{align*}
 \tilde\Theta(z+\zeta/N^{1/2},z+\zeta/N^{1/2})=e^{2|\zeta|^2}\lambda_{\max}(\widetilde{\mathcal{A}})^{N-1}+o(1),
\end{align*}
where $\widetilde{\mathcal{A}}$ is an operator with the kernel
\[ e^{f_{\zeta,+}(R_1-\epsilon \mathcal{M}_0)}B(R_1-R_2)
\mathcal{Z}(R_1,R_2)
e^{f_{\zeta,+}(R_2-\epsilon \mathcal{M}_0)},\quad \mathcal{M}_0=-\frac{1}{2u_*^2}(z\bar \zeta +\bar z\zeta),
\]
where (cf (\ref{pert})) and (\ref{ti-f}))
\[
f_{\zeta,+}=-u_*^4\Tr R^2/2W+\phi_0W^{-3/2}\Tr R^3+ |\zeta|^2/N+   o(N^{-1}).
\]
Here $\phi_0$ is some constant not important for us. 

 Using (\ref{diff_Z}), we can replace $\mathcal{Z}(R_1,R_2)$ by $\mathcal{Z}(R_1-\epsilon \mathcal{M}_0,R_2-\epsilon \mathcal{M}_0)$ with an error $O(W^{-2})$. Then, changing the variables $R_1-\epsilon \mathcal{M}_0\to R_1$ and $R_2-\epsilon\mathcal{M}_0\to R_2$, we obtain by (\ref{diff_eig})  
 \[
 \lambda_{\max}(\widetilde{\mathcal{A}})=e^{2|\zeta|^2}\lambda_{\max}(\mathcal{A})(1+o(N^{-1}))=e^{2|\zeta|^2}(1+o(N^{-1})).
 \]
$\square$

In the next two lemmas we prove that we can replace $\mathcal{K}_{\zeta}$ in (\ref{Theta1}) by its projection onto the space which we can control with Lemmas \ref{l:Z_0}-\ref{c:1}.

Set 
\begin{align}\label{bbK}
\mathbb{K}=\hat{\mathcal{E}}_{2M}\mathcal{K}_\zeta \hat{\mathcal{E}}_{2M},
\end{align}
where $\hat{\mathcal{E}}_{2M}$ was defined in (\ref{E^ell}), and $M$ was defined in (\ref{bb_K_0}) with some  sufficiently big  $C_0$. 
\begin{lemma}\label{l:K_L} 
If $W,N\to\infty$ in such a way that $W\ge N^{\varepsilon _0}$ 
 with some $\varepsilon_0>0$, then we have
\begin{align}\label{K_L.1}
\tilde\Theta(z_1,z_2)=&(\mathbb{K}^{N-1}g_0,g_0)+o(\lambda_{\max}^{N-1}),\quad g_0=e^{-u_*^2\Tr R^2}.
\end{align}
\end{lemma}

\textit{Proof of Lemma \ref{l:K_L}}. 
We start from the proof of the inequality
\begin{align}\label{n.1}
\|\mathcal{K}_\zeta\|\le \lambda_{\max}(1+k_0/2N).
\end{align}
Recall that the operator of multiplication
by $F_\zeta(R,U)$ has the form (\ref{f_1}). Observe that the remainder in (\ref{f_1}) satisfies the bound
\[
\epsilon^2W^{-3/2}=N^{-1}W^{-1/2}\ll N^{-1}.
\]
Hence, it is sufficient to prove (\ref{n.1}) for  the operator 
$\widetilde{\mathcal{K}}_{\zeta}$ which  corresponds to $\mathcal{K}_{\zeta}$ with  $F_\zeta(R,U)$ replaced by $F_0(R-\epsilon\mathcal{M})(1+f_1(R,U))$
since
\begin{align}\label{n.1.0}
\widetilde{\mathcal{K}}_{\zeta}-\mathcal{K}_{\zeta}=O(\epsilon^2W^{-3/2}).
\end{align}

Notice that if $\zeta=0$, then for each $\ell=0,1,\dots$ the space
$\mathcal{H}\otimes\mathcal{E}^{(\ell)}$ is invariant with respect to $\mathcal{K}_{0}$.
Moreover,  since multiplication by $f_1$ can transform $h\in \mathcal{E}^{(\ell)}$ into a function which has nonzero components only in 
$\mathcal{E}^{(\ell-1)},\mathcal{E}^{(\ell)},\mathcal{E}^{(\ell+1)}$, the matrix $\tilde{\mathcal{K}}_{\zeta,L}$ 
is ``block three-diagonal"  in the basis of $\mathcal{H}_{L}\otimes\mathcal{E}^{(\ell)}$. 
Set 
\[
\widetilde{\mathcal{K}}_{\zeta}^{(\ell\ell')}=\mathcal{E}^{(\ell)}\widetilde{\mathcal{K}}_{\zeta}\mathcal{E}^{(\ell')},
\]
and take $M$ defined by (\ref{bb_K_0}).
%\[M'=\max\{C_0(\epsilon W)^{1/2},\log W\},\]
We apply Proposition \ref{p:norm} to the matrix $\mathbb{M}=\widetilde{\mathcal{K}}_{0}^{(\ell\ell)}$  considered as a block matrix 
with 
\[
\mathbb{M}^{(11)}=\mathcal{P}_L\mathbb{M}\mathcal{P}_L,\quad \mathbb{M}^{(12)}=\mathcal{P}_L\mathbb{M}(1-\mathcal{P}_L),
\quad \mathbb{M}^{(22)}=(1-\mathcal{P}_L)\mathbb{M}(1-\mathcal{P}_L)
\]
with $\mathcal{P}_L$ of (\ref{P_L}). We use the bounds
\begin{align*}
&\mathbb{M}^{(11)}\le 1-C_1(\ell/W)^2+C_2(\ell/W)(L/W)<1-C_1(\ell/W)^2/2,\,(M\le \ell\le 2M),\\ & \|\mathbb{M}^{(12)}\|\le CW^{-3/2},
\quad \mathbb{M}^{(11)}\le 1-CL/W,
\notag\end{align*}
where the first one follows from (\ref{KPsi.1}), the second -- from (\ref{exp_dec}), and the last one -- from Lemma \ref{l:2}. Then
 we get
\begin{align}\label{b_K_ll}
\widetilde{\mathcal{K}}_{\zeta}^{(\ell\ell)}\le 1-C'(\ell/W)^2.
\end{align}
Thus, since $\|\widetilde{\mathcal{K}}_{\zeta}^{(\ell\ell+1)}\|\le C(\epsilon/W)$, we have for $\ell>M$ (assuming $C_0$  in (\ref{bb_K_0}) is sufficiently big):
\begin{align*}
\|(\widetilde{\mathcal{K}}_{\zeta}^{(\ell\ell)}-z)^{-1}\widetilde{\mathcal{K}}_{\zeta}^{(\ell\ell+1)}\| \le C(\epsilon/W)(M/W)^{-2}\le q/2
\end{align*}
with some small enough fixed $q$ ($q^{\log W}<N^{-3}$). Here and below in the proof we take $|z|>\lambda_{\max}(1+k/N)$.

Hence, if  we denote by $\widetilde{\mathcal{K}}_{\zeta,M_1,M_2}$   the  block of  $\widetilde{\mathcal{K}}_{\zeta}$ corresponding
to all $\ell,\ell'\in [M_1,M_2)$, then, denoting by $D$ and $D^{(off)}$ the diagonal and off diagonal parts part of $(\widetilde{\mathcal{K}}_{\zeta,M+1,2M+1}-z)$ respectively,
we  get
\begin{align*}\notag
\Big((\widetilde{\mathcal{K}}_{\zeta,M+1,2M+1}-z)^{-1}\Big)^{(\ell,\ell')}=&
(D^{-1/2})^{(\ell\ell)}\big((1+D^{-1/2}D^{(off)}D^{-1/2})^{-1}\big)^{(\ell\ell')}(D^{-1/2})^{(\ell'\ell')}\\
=&(D^{-1})^{(\ell\ell)}\sum_{p\ge|\ell-\ell'|}\big((-D^{(off)}D^{-1})^{p}\big)^{(\ell\ell')},
\end{align*}
and thus bounds above yield 
\begin{align}
\Big\|\Big((\widetilde{\mathcal{K}}_{\zeta,M+1,2M+1}-z)^{-1}\Big)^{(\ell,\ell')}\Big\|\le CNq^{|\ell-\ell'|}.
\label{exp_dec.2}
\end{align}
Here we used $(M/W)^{-2}\le C\sqrt{WN}\le CN$.

By the inversion formula for a block matrix, to prove (\ref{n.1}) it is sufficient to prove that there exists $k>0$ such that for $ |z|>\lambda_{\max}(1+k/N)$ the matrix
\begin{align}\label{n.3}
\widetilde{\mathcal{K}}_{\zeta,0,M+1}-z-\widetilde{\mathcal{K}}_{\zeta}^{(M,M+1)}\Big((\widetilde{\mathcal{K}}_{\zeta,M+1,\infty}-z)^{-1}\Big)^{(M+1,M+1)}
\widetilde{\mathcal{K}}_{\zeta}^{(M+1,M)}
\end{align}
is invertible. But introducing a block diagonal matrix
\[
\widehat{\mathcal{K}}_{\zeta,M+1,\infty}=\mathrm{diag}\{  \widetilde{\mathcal{K}}_{\zeta,M+1,2M+1},\widetilde{\mathcal{K}}_{\zeta,2M+1,\infty}   \}
\]
and using the resolvent identity for  the resolvents of $\widetilde{\mathcal{K}}_{\zeta,M+1,\infty}$ and of $\widehat{\mathcal{K}}_{\zeta,M+1,\infty}$, we obtain by (\ref{exp_dec.2})
\begin{align*}
&\Big((\widetilde{\mathcal{K}}_{\zeta,M+1,\infty}-z)^{-1}\Big)^{(M+1,M+1)}=\Big((\widetilde{\mathcal{K}}_{\zeta,M+1,2M+1}-z)^{-1}\Big)^{(M+1,M+1)}\\&+
\Big((\widetilde{\mathcal{K}}_{\zeta,M+1,2M+1}-z)^{-1}\Big)^{(M+1,2M)}\widetilde{\mathcal{K}}_{\zeta}^{(2M,2M+1)}
\Big((\widetilde{\mathcal{K}}_{\zeta,M+1,\infty}-z)^{-1}\Big)^{(2M+1,M+1)}\\&=
\Big((\widetilde{\mathcal{K}}_{\zeta,M+1,2M+1}-z)^{-1}\Big)^{(M+1,M+1)}+o(N^{-1}).
\end{align*}
Hence, if we prove that for $|z|>\lambda_{\max}(1+k/N)$ the matrix
\begin{align}\label{n.4}
\widetilde{\mathcal{K}}_{\zeta,0,M+1}-z-\widetilde{\mathcal{K}}_{\zeta}^{(M,M+1)}(\widetilde{\mathcal{K}}_{\zeta,M+1,2M+1}-z)^{-1} 
\widetilde{\mathcal{K}}_{\zeta}^{(M+1,M)}
\end{align}
is invertible, then for $|z|>\lambda_{\max}(1+2k/N)$ the matrix in (\ref{n.3}) is invertible, and  thus get (\ref{n.1}). But the the inverse of the matrix (\ref{n.4}) corresponds
to the left upper block of the resolvent of $\widetilde{\mathcal{K}}_{\zeta,0,2M+1}$, hence, it is sufficient to prove that
\begin{align}\label{n.5}
\|\widetilde{\mathcal{K}}_{\zeta,0,2M+1}\|<\lambda_{\max}(1+k/N)
\end{align}
with some $k$.

Consider $\widetilde{\mathcal{K}}_{\zeta,0,2M+1}$ as a block matrix with 
\[
\widetilde{\mathcal{K}}_{\zeta}^{(11)}=\mathcal{P}_L\widetilde{\mathcal{K}}_{\zeta,0,2M+1}\mathcal{P}_L,\quad
\widetilde{\mathcal{K}}_{\zeta}^{(22)}=(I-\mathcal{P}_L)\widetilde{\mathcal{K}}_{\zeta,0,2M+1}(I-\mathcal{P}_L),\quad
\widetilde{\mathcal{K}}_{\zeta}^{(12)}=(I-\mathcal{P}_L)\widetilde{\mathcal{K}}_{\zeta,0,2M+1}\mathcal{P}_L,
\]
with $\mathcal{P}_L$ of (\ref{P_L}).
%and denote $\lambda_{\max,\zeta}=\lambda_{\max} (\widetilde{\mathcal{K}}_{\zeta})$.  
Then by Lemma \ref{l:2} and (\ref{exp_dec})
\[
\widetilde{\mathcal{K}}_{\zeta}^{(22)}<1-CL/W,\quad
 \|\widetilde{\mathcal{K}}_{\zeta}^{(12)}\|
\le C(W^{-3/2}+\epsilon/W).
\]
Hence, for
\[ \widetilde{\mathcal{M}}(z)=\widetilde{\mathcal{K}}_{\zeta}^{(12)}(\widetilde{\mathcal{K}}_{\zeta}^{(22)}-z)^{-1}\widetilde{\mathcal{K}}_{\zeta}^{(21)}
\]
we have
\begin{align}\label{n.2}
\| \widetilde{\mathcal{M}}(z)\|\le CL^{-1}(W^{-2}+N^{-1}).
\end{align}
Moreover, 
since (\ref{exp_dec}) implies that $\mathcal{K}_{\zeta,(\bar k,\bar j)}$ decays as $W^{-|\bar k-\bar j|/2} $,
we get  that there exists fixed $p>0$ such that we can consider  $\mathcal{K}_{\zeta,(\bar k,\bar j)}$ as $2p+1$ block diagonal matrix with an error
$O(N^{-2})$. Hence, $\widetilde{\mathcal{K}}_{\zeta}^{(12)}$ (with an error $O(N^{-2})$) can be considered as a matrix which contains only $p$ 
nonzero diagonals in the bottom left corner, and $\widetilde{\mathcal{K}}_{\zeta}^{(21)}$ can be considered as a matrix which contains only $p$ 
nonzero diagonals in the top right corner. Thus, $\widetilde{\mathcal{M}}(z)$ (with an error $O(WN^{-2})$)  is a matrix which has nonzero component
only in the $p\times p$ block in the bottom right corner, or
\begin{align}\label{n.6}
\widetilde{\mathcal{M}}(z)=&\widetilde{\mathcal{M}}_1(z)+O(WN^{-2}),\quad \widetilde{\mathcal{M}}_1(z)=\sum_{|\bar j|,|\bar k|=L-p}^Lm_{\bar j,\bar k}\Psi_{*\bar j}\otimes\Psi_{*\bar k}.
\end{align}
Consider now the vectors $\{\tilde\Psi_{\ell,k}\}$ introduced in (\ref{Psi^p}).
Denote by $\mathfrak{P}^{(1)}_{\epsilon}$ the orthogonal projection on 
$\mathrm{Lin}\{\tilde\Psi_{\ell,k,\epsilon}\}_{l\le M, |k|\le l}$ (see (\ref{Psi_eps})),
 define $\mathfrak{P}^{(2)}_{\epsilon}=1-\mathfrak{P}^{(1)}_{\epsilon}$, and set
\[ 
\mathbb{M}^{(\alpha\beta)}=\mathfrak{P}^{(\alpha)}_{\epsilon}(\mathcal{K}_{\zeta}^{(11)}-\widetilde{\mathcal{M}})\mathfrak{P}^{(\beta)}_{\epsilon},\quad \alpha,\beta=1,2.
\]
By (\ref{de_Psi}) and  (\ref{Psi^p}),  $\tilde\Psi_{\ell,k}$ has nonzero components only with respect to $\Psi_{*\bar k}$  with $|\bar k|\le p'$, where $p'$ is sufficiently big but fixed
number.
Expanding $\Psi_{\ell,k,\epsilon}$ with respect to $\epsilon$, one can see  that  $\Psi_{\ell,k,\epsilon}$ has nonzero components only with respect to $\Psi_{*\bar k}$  with $|\bar k|\le p'+p''$ plus $O(N^{-2})$ term. Here we chose $p''$ sufficiently big  to have $\epsilon^{p''}\le N^{-2}$. Thus
(\ref{n.6}) yields

\begin{align*}
&\widetilde{\mathcal{M}}_1\tilde\Psi_{\ell,k,\epsilon}=O(N^{-2})\\
\Rightarrow &\mathfrak{P}^{(1)}_{\epsilon}\widetilde{\mathcal{M}}\mathfrak{P}^{(1)}_{\epsilon}=O(WN^{-2}),\quad
\mathfrak{P}^{(1)}_{\epsilon}\widetilde{\mathcal{M}}\mathfrak{P}^{(2)}_{\epsilon}=O(WN^{-2})\quad\Rightarrow\quad \|\mathbb{M}^{(12)}\|\le C( \epsilon/W).
\end{align*}

Moreover,   Lemma \ref{l:la_max} implies that $\mathcal{K}_{\zeta}^{(11)}\Big|_{\zeta=0}=\mathbb{K}_0$ has eigenvalues $\{\lambda_{\ell,k}\}$ (corresponding
to $\{\tilde\Psi_{\ell,k}\}$) in the $c(M/W)^2$-neighbourhood of $\lambda_{\max}$, and all other eigenvalues are less than $\lambda_{\max}-c/W$. Therefore,
\begin{align*}
&\mathfrak{P}^{(2)}_{\epsilon}\mathbb{K}_{0}\mathfrak{P}^{(2)}_{\epsilon}\le \lambda_{\max}-c/W\Rightarrow
\mathfrak{P}^{(2)}_{\epsilon}\mathcal{K}_{\zeta}^{(11)}\mathfrak{P}^{(2)}_{\epsilon}\le \lambda_{\max}-c/W+C\epsilon/W\\ &\Rightarrow
\mathbb{M}^{(22)}\le \lambda_{\max}-c/2W.
\end{align*}
Thus, to prove (\ref{n.1})  it is sufficient to prove that
\[
\mathbb{M}^{(11)}\le \lambda_{\max}(1+k/N).
\]
This can be done by applying Proposition \ref{p:norm} to $\tilde{\mathbb{M}}=\mathbb{M}^{(11)}$  with blocks
\[\tilde{\mathbb{M}}^{(22)}=(\mathcal{K}_\zeta \Psi_{0,0,\epsilon},\Psi_{0,0,\epsilon}),\quad \tilde{\mathbb{M}}^{(21)}=\tilde{\mathbb{M}}\Psi_{0,0,\epsilon},\quad \delta=1/N,
\]
if we use  (\ref{cor:1}) and (\ref{cor:1.2}).

\medskip

To prove (\ref{K_L.1}) we observe first that, expanding $F_\zeta(R,U)$ in the series with respect to $\epsilon$, 
one can replace $F_\zeta(R,U)$ in $\mathcal{A}_{\zeta}$ by $F_0(R,U)(1+f_2(R,U,\epsilon))$ (see (\ref{A_zeta}) and (\ref{M(U)}))  such that $f_2$ includes terms containing $L_{U^*}$ only a finite 
number of times (we call it $s$). Denote the operator with this new $\mathcal{A}_{\zeta}$ by $\widetilde{\mathcal{K}}_{\zeta}'$, then we can choose $s$ big enough such that
\begin{align}\label{K-K'}
\|\widetilde{\mathcal{K}}_{\zeta}'-\mathcal{K}_{\zeta}\|\ll N^{-1}W^{-2}.
\end{align}
Since $\widetilde{\mathcal{K}}_{\zeta}'$  is $2s+1$-diagonal matrix, we can repeat the argument used above for $\widetilde{\mathcal{K}}_{\zeta}$
(with may be bigger $C,C_0$ in the definition of $M$ in (\ref{bb_K_0})), and get
\begin{align}\label{K-K_M}
\|(\widetilde{\mathcal{K}}_{\zeta}'-z)^{-1}-(\widetilde{\mathcal{K}}'_{\zeta,2M+1,\infty}-z)^{-1}\|\le N^{-3},\quad |z|>\lambda_{\max}(1+k_0/N).
\end{align}
Then, using the Cauchy residue theorem  and (\ref{K-K'}), one can obtain for $\omega=\{z: |z|=\lambda_{\max}(1+2k_0/N)\}$
\begin{align*}
|(\mathcal{K}_{\zeta}^{N-1}g,g)-&((\widetilde{\mathcal{K}}_{\zeta}')^{N-1}g,g)|=C\Big|\oint\limits_{\omega}  z^{N-1}
\Big((\widetilde{\mathcal{K}}_{\zeta}'-z)^{-1}({\widetilde{\mathcal{K}}_{\zeta}'-\mathcal{K}}_{\zeta})(\mathcal{K}_{\zeta}-z)^{-1}g,g\Big)dz\Big|\\
\le &C \lambda_{\max}^{N} \|\widetilde{\mathcal{K}}_{\zeta}'-\mathcal{K}_{\zeta}\|\|g\|^2\oint\limits_{\omega} |dz| |z-\lambda_{\max}-k_0/N|^{-2}=o(\lambda_{\max}^{N}).
\end{align*}
Here we used that $\|g\|\le CW$ and that for any matrix $\mathbb{M}:\mathbb{M}=\mathbb{M}^*$, $\|\mathbb{M}\|\le \lambda_{\max}(1+k_0/N)$
\begin{align*}
\|(\mathbb{M}-z)^{-1}\|\le C|z- \lambda_{\max}(1+k_0/N)|^{-1}.
\end{align*}
Similarly, from (\ref{K-K_M}) we get
\begin{align*}
|((\widetilde{\mathcal{K}}_{\zeta}')^{N-1}g,g)-((\widetilde{\mathcal{K}}'_{\zeta,2M+1,\infty})^{N-1}g,g)|=o(\lambda_{\max}^{N}),
\end{align*}
and (\ref{K-K'}) yields
\[
|((\hat{\mathcal{E}}_{2M}\mathcal{K}_{\zeta}\hat{\mathcal{E}}_{2M})^{N-1}g,g)-((\widetilde{\mathcal{K}}'_{\zeta,2M+1,\infty})^{N-1}g,g)|=o(\lambda_{\max}^{N}).
\]
The last three bounds imply (\ref{K_L.1}).

$\square$

\begin{lemma}\label{l:P} Denote
$\mathfrak{P}_{\epsilon}^{(1)}$ the orthogonal projection on the subspace $\mathrm{Lin}\{\Psi_{\ell,k,\epsilon}\}_{\ell\le M,|k|\le \ell}$
 defined by (\ref{Psi_eps}) for $\Psi_{\ell,k}$ of (\ref{de_Psi}). 
  Then
\begin{align}\label{l:P.1}
(\mathbb{K}^{N-1}g_0,g_0)=
((\mathfrak{P}_{\epsilon}^{(1)}\mathbb{K}\mathfrak{P}_{\epsilon}^{(1)})^{N-1} g_1,g_1)+o(1),\quad g_1=\mathfrak{P}_{\epsilon}^{(1)}g_0.
\end{align}
\end{lemma}
Notice that in contrast to $g_0$ with $\|g_0\|=CW$, by (\ref{de_Psi}) we have that
\begin{align}\label{normg_1}
&\mathrm{Lin}\{\tilde\Psi_{\ell,k,\epsilon}\}\in\mathrm{Lin}\{\Psi_{*\bar j}\}_{|\bar j|\le p}\otimes L_2(U(2))\\
\Rightarrow&\|g_1\|^2\le \sum_{|\bar j|\le p}\int dRdR'\Psi_{*\bar j}(R)\Psi_{*\bar j}(R')g(R,U)g(R',U)dU\le C
%=&\Big|\int \Psi_{0,0}(R)e^{-u_*^2\Tr R^2/W}dRdU\Big|^2+O(\epsilon \log ^2 W/W)=O(1),\notag
\notag\end{align}
%where $g_\epsilon(R,U)=g_0(R+\epsilon\mathcal{M}(U))$.

\textit{Proof of Lemma \ref{l:P}.} 
We prove first that 
\begin{align}\label{ti-g}
(\mathbb{K}^{N-1}g_0,g_0)=(\mathbb{K}^{N-1}\tilde g,\tilde g)+o(\lambda_{\max}^{N-1}),\quad \tilde g=\mathcal{P}_{\bar 0}g,\quad\|\tilde g\|^2\le C,
\end{align}
where $\mathcal{P}_{\bar 0}$ is an orthogonal projection on the space $\{\Psi_{\bar 0}(R)h(U)\}_{h\in \mathcal{E}^{(M)}}$ with $\Psi_{\bar 0}$ corresponding to
$\lambda_{\max}(\mathcal{A})$.

Consider $\mathbb{K}$ as a block matrix with
\begin{align*}
&\mathbb{K}^{(00)}=\mathcal{P}_{\bar 0}\mathbb{K}\mathcal{P}_{\bar 0},\quad \mathbb{K}^{(11)}=(1-\mathcal{P}_{\bar 0})\mathbb{K}(1-\mathcal{P}_{\bar 0}), \\
 &\mathbb{K}^{(01)}=\mathcal{P}_{\bar 0}\mathbb{K}(1-\mathcal{P}_{\bar 0}), \quad  \mathbb{K}^{(10)}=(1-\mathcal{P}_{\bar 0})\mathbb{K}\mathcal{P}_{\bar 0}.
\end{align*}
Then,  since for $\mathbb{K}_0$ of (\ref{bb_K_0})
\begin{align}\label{diff_K}
\|\mathbb{K}-\mathbb{K}_{0}\|\le C\epsilon/W, 
\end{align}
by Lemma \ref{l:la_max} we conclude that $\mathbb{K}^{(00)}$ has $(M+1)^2$ eigenvalues 
$\lambda$ on the distance less than $C(\epsilon/W)$ from $\lambda_{\max}$.
Moreover, since  we proved in  Lemma \ref{l:la_max} that all remaining eigenvalues of $\mathbb{K}_0$ are less than
$ \lambda_{\max}-c/W$,  (\ref{diff_K}) yields also that
%\begin{align*}
%\mathbb{K}^{(11)}\le \lambda_{\max}-c/W, \quad \|\mathbb{K}^{(01)}\|\le C(W^{-3/2}+\epsilon/W)
%\end{align*}
all the remaining eigenvalues of $\mathbb{K}$ are less than $\lambda_{\max}-c/2W$. 

Denote $\mathbb{E}_0$ the spectral projection of $\mathbb{K}$ on the subspace
spanned on the $\{\Phi_\lambda\}_{|\lambda-\lambda_{\max}|\le C(\epsilon/W)}$, where
$\{\Phi_\lambda\}$ are  eigenvectors, corresponding to   the
first $(M+1)^2$ eigenvalues of $\mathbb{K}$. Then
\[
(\mathbb{K}^{N-1}g_0,g_0)=(\mathbb{K}^{N-1}\mathbb{E}_0g_0,\mathbb{E}_0g_0)+O(\lambda_{\max}^Ne^{-cN/2W}\|g_0\|^2)
=(\mathbb{K}^{N-1}\mathbb{E}_0g_0,\mathbb{E}_0g_0)+o(\lambda_{\max}^N).
\]
Consider any  $\Phi_\lambda$ corresponding to  one  of the
first  $(M+1)^2$ eigenvalues of $\mathbb{K}$, and introduce
\[
\Phi_\lambda^{(0)}=\mathcal{P}_{\bar 0}\Phi_\lambda, \quad 
\Phi_\lambda^{(1)}=(1-\mathcal{P}_{\bar 0})\Phi_\lambda.
\]
Then it follows from the equation $(\mathbb{K}-\lambda)\Phi=0$ that
\begin{align}\notag
&\mathbb{K}^{(10)}\Phi_\lambda^{(0)}+(\mathbb{K}^{(11)}-\lambda)\Phi_\lambda^{(1)}=0\Rightarrow 
\Phi_\lambda^{(1)}=-(\mathbb{K}^{(11)}-\lambda)^{-1}\mathbb{K}^{(10)}\Phi_\lambda^{(0)}.
%\\
%&\Rightarrow \|\Phi_\lambda^{(2)}\|\le Wc^{-1}\|\mathbb{K}^{(12)}\|\\ \Rightarrow&
%\Big\| (\mathbb{E}_{\rho}-\mathfrak{P}_{\epsilon}^{(1)})g_0\|^2\le \sum_{\lambda}|(g_0,\Phi_\lambda-\Phi_\lambda^{(1)}))|^2
%CW^2 \|\mathbb{K}^{(12)}\|\,\mathrm{dim}(1-\mathbb{E}_{1-\rho})=o(1).
%\label{l:P.4}
\end{align}
Set
\[\mathbb{K}_{(\bar j,\bar j')}=\mathcal{P}_{\bar j}\mathbb{K}\mathcal{P}_{\bar j'},\quad g_{\bar j}=\mathcal{P}_{\bar j}g,
\]
where $\mathcal{P}_{\bar j}$ is an orthogonal projection on $\mathrm{Lin}\{\Psi_{*\bar j}h(U)\}_{h\in\mathcal{E}_M}$.

Repeating almost literally the argument of Lemma \ref{l:la_max}, we obtain an analogue of (\ref{exp_dec}):
\begin{align}\label{exp_dec.1}
&\|\mathbb{K}_{(\bar j,\bar j')}\|\le C\min\{W^{-3/2},W^{-|\bar j-\bar j'|/2}\},\quad |\bar j|,|\bar j'|\le L,\quad|\bar j-\bar j'|\not= 0\\
&\|(\mathbb{K}^{(11)}-\lambda)^{-1}_{(\bar j,\bar j')}\|\le  W^{1-|\bar j-\bar j'|/2}.
\notag\end{align}
%$g\in \mathcal{E}^{(0)}$ and 
Hence, it is easy to see that there exists $p'>0$ such that 
\begin{align}\label{normg2}
\Phi_\lambda=\Phi^{(0)}_\lambda+W^{-1/2}\tilde\Phi^{(1)}_\lambda+O(N^{-1}),\quad 
\mathrm{Lin}\{\Phi^{(1)}_\lambda\}_{\lambda}\subset\mathrm{Lin}\{\Psi_{*\bar j}\}_{|\bar j|\le p}\otimes L_2(U(2))
\end{align}
Then, repeating (\ref{normg_1}), we obtain (\ref{ti-g}) for $\tilde g=\mathcal{P}_{\bar 0}g$.
%\begin{align*}
%(g_0,\Phi_\lambda)=&(g_0,\Phi^{(0)}_\lambda)-\sum_{|\bar j|\ge 1,|\bar j'|\ge 1} \Big(g_{\bar j},(\mathbb{K}^{(11)}-\lambda)^{-1}_{(\bar j,\bar j')}
%\mathbb{K}_{(\bar j',\bar 0)}\Phi^{(0)}_\lambda\Big)\\=&(g_0,\Phi^{(0)}_\lambda)+O(W^{-1/2}),\\
%\Rightarrow \mathbb{E}_0g_0=&\sum_\lambda (g_0,\Phi_\lambda)\Phi_\lambda=\sum_\lambda  (g_0,\Phi^{(0)}_\lambda)\Phi^{(0)}_\lambda+O(W^{-1/2}\log^2W )\\
%%=&
%%\mathcal{P}_{\bar 0}g_0+O(W^{-1/2}\log^2W ),\\
%\Rightarrow  \|\tilde g\|^2=&\sum_{\lambda,\lambda'}(g_0,\Phi_\lambda^{(0)})(\Phi_\lambda^{(0)},\Phi_{\lambda'}^{(0)})(\Phi_{\lambda'}^{(0)},g_0)
%+O(W^{-1/2}\log^4W )\\
%=&\|\mathcal{P}_{\bar 0}g_0\|^2 +O(W^{-1/2}\log^4W )\le C
%%+O(W^{-1/2}\log^2W )\le C
%\end{align*}
%for $\tilde g=\mathcal{P}_{\bar 0}g$. This completes the proof of .
:

Now let us prove (\ref{l:P.1}).  Set  $\mathfrak{P}_{\epsilon}^{(2)}=I-\mathfrak{P}_{\epsilon}^{(1)}$ and
\begin{align}\label{K^ab.0}
\mathbb{K}^{(\alpha\beta)}=&\mathfrak{P}_{\epsilon}^{(\alpha)}\mathbb{K}\mathfrak{P}_{\epsilon}^{(\beta)},\quad \alpha,\beta=1,2.
\end{align}
%Evidently
% \[ \mathfrak{P}_{\epsilon}^{(\alpha)}\mathfrak{P}_{\epsilon}^{(\beta)}=0,\quad \alpha\not=\beta.
% \]
Then we have  the bounds
\begin{align}\label{b_K.0}
&\| \mathbb{K}^{(11)}-\lambda_{\max} \|\le C(\epsilon/W+N^{-1}),\\
 & \|\mathbb{K}^{(12)}\|\le C \epsilon W^{-3/2},\quad
\mathbb{K}^{(22)}\le 1-c_0/W.
\notag\end{align}
The first bound here follow from Lemma \ref{l:la_max} and (\ref{cor:1.2}), the second -- from (\ref{cor:1.0}), and the last  bound was proved in Lemma \ref{l:K_L}.

Since we have proved above that $\mathbb{K}$ has $(M+1)^2$ eigenvalues in the $(\epsilon/W)$-neighbourhood of $\lambda_{\max}$ and all
the remaining eigenvalues are less that $1-c/2W$, and we also have (\ref{n.1}), we can 
apply the Cauchy residue theorem in the following  form:
\begin{align*}
(\mathbb{K}^{N-1}\tilde g,\tilde g)=&-\frac{1}{2\pi i}\Big(\oint_{\mathcal{L}}+\oint_{|z|=1-c/2W}\Big)
z^{N-1}\sum_{\alpha,\beta=1}^2\Big(\mathbb{G}^{(\alpha\beta)}(z)\mathfrak{P}_{\epsilon}^{(\beta)}\tilde g,\mathfrak{P}_{\epsilon}^{(\alpha)}\tilde g\Big)dz\\
%+ \oint_{|z|=1-c/W}\Big)z^{N-1}\\
=&-\frac{1}{2\pi i}\oint_{\mathcal{L}}z^{N-1}
\sum_{\alpha,\beta=1}^2\Big(\mathbb{G}^{(\alpha\beta)}(z)\mathfrak{P}_{\epsilon}^{(\beta)}\tilde g,\mathfrak{P}_{\epsilon}^{(\alpha)}\tilde g\Big)dz+o(\lambda_{\max}^{N-1}).\end{align*}
Here $\mathbb{G}(z)=(\mathbb{K}-z)^{-1}$ and
\begin{align*}
\mathcal{L}=&\partial \Omega,\quad \Omega=\{z:|z|\le \lambda_{\max}(1+2k_0/N)\wedge|z-\lambda_{\max}|\le C(\epsilon/W)\}
%\mathbb{G}^{(11)}(z)=\big(\mathbb{K}^{(11)}-z-\mathbb{K}^{(12)}(\mathbb{K}^{(22)}-z)^{-1}\mathbb{K}^{(21)}\big)^{-1},
\end{align*}

Since the spectrum of $\mathbb{K}$ belongs to $[0,\lambda_{\max}(1+k_0/N)]$ (see (\ref{n.1})), by (\ref{b_K.0}) and the standard resolvent bounds we have for $z\in \mathcal{L}$
\begin{align*}
&\|\mathbb{G}^{(11)}(z)\|,\|(\mathbb{K}^{(11)}-z)^{-1}\| \le C|z-\lambda_{\max}(1+k_0/N)|^{-1},
\qquad \qquad\|\mathbb{G}^{(22)}\|\le CW,\\
%\|(\mathbb{K}^{(11)}-z)^{-1}\|\le C|z-\lambda_{\max}(1+k_0/N)|^{-1},\\
&\|\mathbb{G}^{(12)}\|= \|(\mathbb{K}^{(11)}-z)^{-1}\mathbb{K}^{(12)}\mathbb{G}^{(22)}\|\le C|z-\lambda_{\max}(1+k_0/N)|^{-1}\epsilon/W^{1/2}.
\end{align*}
Hence, we conclude that the integrals with $\mathbb{G}^{(12)}$ and $\mathbb{G}^{(21)}$ gives us $o(\lambda_{\max}^{N-1})$.
In addition, using (\ref{b_K.0}) and the last bound, we obtain
\begin{align*}
&\Big|\oint_{\mathcal{L}} z^{N-1}
\Big(\big(\mathbb{G}^{(11)}(z)-(\mathbb{K}^{(11)}-z)^{-1}\big)\mathfrak{P}_{\epsilon}^{(1)}\tilde g,\mathfrak{P}_{\epsilon}^{(1)}\tilde g\Big)dz \Big|\\
&\le C \|\mathbb{K}^{(21)}\|^2\|\tilde g\|^2\sup_z\|(\mathbb{K}^{(22)}-z)^{-1}\|\cdot \oint_{\mathcal{L}}  \|\mathbb{G}^{(11)}(z)\|\cdot \|(\mathbb{K}^{(11)}-z)^{-1}\||dz| 
\\
&\le C(\epsilon^2/W^3)\cdot W\cdot N=C/W=o(1),\\
&\Big|\oint_{\mathcal{L}} z^{N-1}
\Big(\big(\mathbb{G}^{(22)}(z)-(\mathbb{K}^{(22)}-z)^{-1}\big)\mathfrak{P}_{\epsilon}^{(2)}\tilde g,\mathfrak{P}_{\epsilon}^{(2)}\tilde g\Big)dz \Big|\\
&\le C \|\mathbb{K}^{(21)}\|^2\|\tilde g\|^2\cdot \sup_z(\|(\mathbb{K}^{(22)}-z)^{-1}\|\cdot \|\mathbb{G}^{(22)}(z)\|) \oint_{\mathcal{L}} \|(\mathbb{K}^{(11)}-z)^{-1}\||dz| \\
&\le C\varepsilon^2/W^3\cdot W^2\cdot \log N=C\log N/N=o(1).
\end{align*}
Hence,
\begin{align*}
(\mathbb{K}^{N-1}\tilde g,\tilde g)=&
-\frac{1}{2\pi i}\oint_{\mathcal{L}}z^{N-1}((\mathbb{K}^{(11)}-z)^{-1}\mathfrak{P}_{\epsilon}^{(1)}\tilde g,\mathfrak{P}_{\epsilon}^{(1)}\tilde g)dz\\
&-\frac{1}{2\pi i}\oint_{\mathcal{L}}z^{N-1}((\mathbb{K}^{(22)}-z)^{-1}\mathfrak{P}_{\epsilon}^{(2)}\tilde g,\mathfrak{P}_{\epsilon}^{(2)}\tilde g)dz
+o(\lambda_{\max}^{N-1}).
\end{align*}
Observe that the second integral here is zero, since $(\mathbb{K}^{(22)}-z)^{-1}$ is analytic in $\Omega$.
Thus, applying the Cauchy residue theorem backward, we obtain (\ref{l:P.1}) with  $g_1=\mathfrak{P}_{\epsilon}^{(1)} g_0$ replaced by 
$\mathfrak{P}_{\epsilon}^{(1)}\mathcal{P}_{\bar 0}g_0$.

But in view  of representations (\ref{normg2}) and (\ref{de_Psi}), we have 
\[\|g_1-\tilde g\|=O(W^{-1/2}).\]
This completes the proof of the lemma.

$\square$

\textit{Poof of Theorem \ref{t:Gin}}. 
By (\ref{cor:1.0})    we have 
\begin{align*}
\mathbb{K}^{(11)} 
\Psi_{\ell,k,\epsilon}=&
\lambda_{\ell}e^{2\nu/N}\Psi_{\ell,k,\epsilon}+O(\epsilon W^{-3/2})=
(1+\mathcal{D})\Psi_{\ell,k,\epsilon}+O(\epsilon W^{-3/2}),\\
\mathcal{D}_{\ell\ell}=&-l(l+1)/8(u_*W)^2+2N^{-1}\hat\nu _{\ell\ell},\quad 
\mathcal{D}_{\ell,\ell+1}=2N^{-1}\hat\nu _{\ell\ell+1}, \quad
\mathcal{D}_{\ell,\ell+k}=0 \quad(|k|\ge 2)
\\
\Rightarrow \mathbb{K}^{(11)} =&I\otimes(I+\mathcal{D})+o(N^{-1}),
%
%\quad \epsilon W^{-3/2}=(WN^{1/2})^{-1}= o(N^{-1}).
\end{align*}
where $\hat \nu$ was defined in (\ref{nu}).
Since 
\[g_1=g_1^{(0)}+O(\epsilon),\quad g_1^{(0)}\in\mathcal{H}_L\otimes\mathcal{E}_0,
\] 
it is sufficient to prove that 
\[ \Big((I+\mathcal{D})^N\Big)_{00}=\Big((I+2N^{-1}\hat\nu)^N\Big)_{00}+o(1),
\]
Choose $M_0=C_0\log W$ with sufficiently big $C_0$. Then for any $|z|>\lambda_{\max}(1+C/N)$ with sufficiently big $C$ and $1\le\ell\le M_0$
\[
|(I+\mathcal{D}-z)_{\ell,\ell}|^{-1}|\mathcal{D}_{\ell,\ell+1}|\le \frac{1}{4}.
\]
Hence, if we consider a matrix $\hat{\mathcal{D}}$ which is obtained from $\mathcal{D}$ by removing the entries $\mathcal{D}_{M_0,M_0+1}$ and
$\mathcal{D}_{M_0+1,M_0}$, then
repeating the argument of Lemma \ref{l:K_L}, we get
\begin{align*}
|(I+\hat{\mathcal{D}}-z)^{-1}_{0M_0}|\le &C4^{-M_0}N\le CN^{-3}.
\end{align*}
Therefore,
\begin{align*}
&(I+\mathcal{D}-z)^{-1}_{00}=(I+\hat{\mathcal{D}}-z)^{-1}_{00}
+(I+\hat{\mathcal{D}}-z)^{-1}_{0M_0}\mathcal{D}_{M_0,M_0+1}(I+\mathcal{D}-z)^{-1}_{M_0+1,0}\\=&
(I+\mathcal{E}_{M_0}\mathcal{D}\mathcal{E}_{M_0}-z)^{-1}_{00}+O(N^{-2}).
\end{align*}
Hence, we can replace $\mathcal{D}$ by $\mathcal{E}_{M_0}\mathcal{D}\mathcal{E}_{M_0}$. But
\[ (\|\mathcal{E}_{M_0}(I+\mathcal{D})\mathcal{E}_{M_0}-\mathcal{E}_{M_0}e^{2\nu/N}\mathcal{E}_{M_0}\|\le M_0^{2}/W^{2}=o(N^{-1}).\]
%
%\begin{align*}
%\mathbb{K}^{(11)}=\mathcal{E}_{M_0}e^{2\nu/N}\mathcal{E}_{M_0}+o(N^{-1})
%\end{align*}
Combining this relation  with (\ref{l:P.1}),  we finish the proof of (\ref{lim_Gin}).

$\square$

\textit{Proof of Theorem \ref{t:locGin}}.
 Denote by $\mathfrak{P}_\epsilon^{(00)}$ the orthogonal projection on 
  the subspace $\mathrm{Lin}\{\Psi_{\bar 0}(R-\epsilon\mathcal{M}(U))\}$ and by $\mathfrak{P}_{\epsilon}^{(01)}$
  the  orthogonal projection on $\mathrm{Lin}\{\Psi_{\ell,k}(R-\epsilon\mathcal{M}(U),U)\}_{1\le \ell \le M, |k|\le l}$. 

 Evidently,
 \[ \mathfrak{P}_{\epsilon}^{(00)}\mathfrak{P}_{\epsilon}^{(01)}=0,\quad 
 \mathfrak{P}_{\epsilon}^{(00)}+\mathfrak{P}_{\epsilon}^{(01)}=\mathfrak{P}_{\epsilon}^{(1)}.
 \]
 Set
\begin{align}\label{K^ab}
 \mathbb{K}_1^{(\alpha\beta)}=&\mathfrak{P}^{(0\alpha)}\mathbb{K}^{(11)}\mathfrak{P}^{(0\beta)},\quad \alpha,\beta=0,1.
\notag\end{align}
Introduce the resolvent 
\[
\mathbb{G}_1=(\mathbb{K}_1-z)^{-1},
\]
and consider the function
\begin{align*}
\Phi(z)=&\mathbb{K}_1^{(00)}-z-\mathbb{K}_1^{(01)}(\mathbb{K}_1^{(11)}-z)^{-1}\mathbb{K}_1^{(10)}.
%-\mathbb{K}^{(02)}\tilde G^{(22)}\mathbb{K}^{(20)}
%\\
%&+\mathbb{K}^{(01)}\tilde G^{(11)}\mathbb{K}^{(12)}(\mathbb{K}^{(22)}-z)^{-1}\mathbb{K}^{(20)}+
%\mathbb{K}^{(02)}(\mathbb{K}^{(22)}-z)^{-1}\mathbb{K}^{(21)}\tilde G^{(11)}\mathbb{K}^{(10)}.
\end{align*}
Relations (\ref{cor:1}) and (\ref{cor:1.2}) imply the bounds
\begin{align}\label{b_K}
\mathbb{K}_1^{(00)}&=\lambda_{\max}+O(N^{-1}(\epsilon+W^{-1/2}) ),\\
\|\mathbb{K}_1^{(01)}\|&\le C\big(N^{-1}+\epsilon W^{-2}\big),\quad 
\mathbb{K}_1^{(11)}\le 1-C/W^2.
\notag
\end{align}
Then, taking  sufficiently big $C_1$ and setting
\begin{align}\label{Ball}
 \mathcal{B}=\Big\{z:|z-\lambda_{\max}|\le C_1
 \Big(|\mathbb{K}_1^{(00)}-\lambda_{\max}|
 +W^2\|\mathbb{K}_1^{(01)} \|^2\Big) \Big\},
\end{align}
  we get  for $z\in\partial \mathcal{B}$
\begin{align*}
&\|(\mathbb{K}_1^{(11)}-z)^{-1}\|\le CW^2 \Rightarrow \|\mathbb{K}_1^{(01)}(\mathbb{K}_1^{(11)}-z)^{-1}\mathbb{K}_1^{(10)}\|
\le CW^2\|\mathbb{K}_1^{(01)}\|^2\\
\Rightarrow&
|\Phi(z)-(\lambda_{\max}-z)|=\| \mathbb{K}_1^{(00)}-\lambda_{\max}-
\mathbb{K}_1^{(01)}(\mathbb{K}_1^{(11)}-z)^{-1}\mathbb{K}_1^{(10)} \|\\
 &\hskip3,3cm \le C\Big(|\mathbb{K}_1^{(00)}-\lambda_{\max}|
 +W^2\|\mathbb{K}_1^{(01)} \|^2\Big) \le|\lambda_{\max}-z|/2,
\end{align*}
and  the Rouche theorem implies that $\Phi(z)$ has exactly one zero in $\mathcal{B}$.
 Then,  taking into account that
  $\Phi(z)=(\mathbb{G}_1^{(00)}(z))^{-1}$, and, therefore, zeros of $\Phi(z)$ are eigenvalues of $\mathbb{K}_1$, we obtain  that 
 $\mathbb{K}_1$  has exactly one eigenvalue inside the circle, i.e.
 \begin{align}\label{d-la}
 |\lambda_{\max}-\lambda_{\max}(\mathbb{K}_1)|\le  C_1N^{-1}(W^{-1/2}+\epsilon+W^2/N).
 \end{align}
 
 Notice, that the same argument yields that $\mathbb{K}_1$  
 has exactly one eigenvalue inside the circle $|z-\lambda_{\max}|\le 2dW^{-2}$ with sufficiently small fixed $d>0$, i.e. the spectral gap
 of $\mathbb{K}_1$ is more than $dW^{-2}$. 
 Hence, we have
\begin{align}\label{K^N_loc}
(\mathbb{K}^{N-1}_1g_1,g_1)=\lambda_{\max}^{N-1}(\mathbb{K}_1)|(g_1,\Psi_{0,\mathbb{K}_1})|^2
\Big(1+O(e^{-dN/W^2})\Big),
\end{align}
where $\Psi_{ 0,\mathbb{K}_1}$ is an eigenvector of $\mathbb{K}_1$ corresponding to $\lambda_{\max}(\mathbb{K}_1)$. 

Using (\ref{de_Psi}),  we obtain 
 \[\|\Psi_{0,\mathbb{K}_1}-\Psi_{\bar 0}+\Psi_{\bar 0}-\Psi_{*0}\|\to 0\,
 \Rightarrow\, 
 (g_1,\Psi_{\bar 0,\mathbb{K}})=(g_,\Psi_{*0})(1+o(1)).
\]
Thus, using (\ref{l:P.1}) and (\ref{d-la}),   we get 
\begin{align*}
\Theta(z_1,z_2)=\lambda_{\max}^{N-1}|(g_1,\Psi_{*\bar 0})|^2\Big(1+o(1)\Big),
\end{align*}
which implies (\ref{lim_locGin}).

$\square$

\section{Appendix}
\textit{Proof of Lemma \ref{l:A}.}
Relations (\ref{psi_k}) can be checked by straightforward computations (see \cite{SS:17}).

For  $\lambda_{*\bar m}=\lambda_*^{m_0+m_1+m_2+m_3}$ ($m_i\le L$) consider $E_{\lambda_{*\bar m}}$ -- the orthogonal  projection
on the eigenspace corresponding to $\lambda_{*\bar m}$. Denote by $\tilde F$  the operator of  multiplication by $\Tr R^3$ and
\[
\tilde A=\tilde F\mathcal{A}_*+\mathcal{A}_*\tilde F, \quad\mathcal{A}_{0}=\mathcal{A}_{*}+c_3W^{-3/2}\tilde A+O(W^{-2}).
\]
 It is easy to see that 
\[ E_{\lambda_{*\bar m}}\tilde AE_{\lambda_{*\bar m}}=0
\]
Hence,  if we consider  $\mathcal{A}_{0}$ as a block matrix with
$\mathcal{A}_{0}^{(11)}= E_{\lambda_{*\bar m}}\mathcal{A}_{0}E_{\lambda_{*\bar m}}$,
then
\begin{align}\label{est1}
&\mathcal{A}_{0}^{(11)}=\lambda_{*\bar m}E_{\lambda_{*\bar m}}+O(W^{-2}), \quad  \mathcal{A}_{0}^{(12)}=O(W^{-3/2}),\\
&\mathcal{A}_{0}^{(21)}=O(W^{-3/2}),\quad \mathcal{A}_{0}^{(22)}=\mathcal{A}_{*}^{(22)}+O(W^{-3/2}).
\notag\end{align}
Since for  $k_0W^{-2}\le |z-\lambda_{*\bar m}|\le cW^{-1}$ with sufficiently big fixed $k_0$ and sufficiently small
$c>0$ we have
 \begin{align*}
 & \|(\mathcal{A}_{0}^{(22)}-z)^{-1}\|= \|(\mathcal{A}_{*}^{(22)}-z+O(W^{-3/2}))^{-1}\|\le C'W\\
\Rightarrow& \| \mathcal{A}_{0}^{(12)}(\mathcal{A}_{0}^{(22)}-z)^{-1}\mathcal{A}_{0}^{(21)}\|\le C''W^{-2},
\end{align*}
we conclude that
\begin{align*}
%%\|(\mathcal{A}_{*}^{(22)}(t)-z)^{-1}\|\le C'W,\quad\\
&\|(G^{(11)}(z))^{-1}\|= \|\mathcal{A}_{0}^{(11)}-z-\mathcal{A}_{0}^{(12)}(\mathcal{A}_{0}^{(22)}-z)^{-1}\mathcal{A}_{0}^{(21)}\|\ge k_0/2W^2.
%|\varphi(z)-(z-\lambda_{*\bar m})^{\gamma}|\le |z-\lambda_{*\bar m}|^{\gamma} |(1+C''/C_0)^{\gamma} -1)|
%<|z-\lambda_{*\bar m}|^{\gamma},
 \end{align*}
 Hence, $\|G^{(11)}(z)\|$ is finite for $k_0W^{-2}\le |z-\lambda_{*\bar m}|\le cW^{-1}$, and so $\mathcal{A}_0$  has no eigenvalues
 in this annulus. On the other hand, the bound from the second line above yields that eigenvalues of $(G^{(11)}(z))^{-1}$ differ from
 eigenvalues of  $\mathcal{A}_{0}^{(11)}-z$ less than $C''W^{-2}$ if $|z-\lambda_{*\bar m}|\le k_0W^{-2}$. This completes the proof of (\ref{diff_eig}).

Since Lemma \ref{l:Z_0} implies  that relations (\ref{est1}) are
valid also for the operator $\mathcal{A}$ of (\ref{hat_A}), we obtain
that (\ref{diff_eig}) is valid also for eigenvalues of $\mathcal{A}$.

The proof of (\ref{de_Psi.0}) repeats almost literally the proof of (\ref{de_Psi}).

To prove (\ref{norm_A}),
consider $\mathcal{A}_\mathcal{M}$ in the basis of eigenvectors of $\mathcal{A}$ as a
a block matrix with the first block corresponding to $\Psi_{\bar 0}$. Then observe that 
\[
(\mathcal{A}_\mathcal{M}\Psi_{\bar 0},\Psi_{\bar 0})=\lambda_{\max}+O(\epsilon^2W^{-1}),\quad
\]
since in view of (\ref{int_2}) the linear with respect to $\epsilon$ term is equal to zero.  Moreover,
\[
\|\mathcal{A}_\mathcal{M}^{(12)}\|\le C\epsilon W^{-1}, \quad \|\mathcal{A}_\mathcal{M}^{(22)}\|\le \lambda_{\max}-C/W+O(\epsilon/W)\le  \lambda_{\max}-C/2W.
\]
Here for the second inequality we used that $\|\mathcal{A}-\mathcal{A}_M\|=O(\epsilon/W)$.

Then,   (\ref{norm_A}) follows from Proposition (\ref{p:norm}).

$\square$

\medskip

\textit{Proof of Proposition \ref{p:as_Leg}.}

 We use the following representations of $P_{k+q,k}^{(\ell)}$ (see \cite{Vil:68})
\begin{align*}
P_{k+q,k}^{(\ell)}(\cos\theta)=&\frac{\mu_{\ell,k,q}}{2\pi }
\int(\cos(\theta/2)+i\sin(\theta/2)e^{i\phi})^{\ell+k}
(\cos(\theta)+i\sin(\theta/2)e^{-i\phi})^{\ell-k}e^{iq\phi}d\phi,\\
=&\frac{\mu_{\ell,k,q}}{2\pi }\int\cos^{2\ell}(\theta/2)(1+i\tan(\theta/2)e^{i\phi})^{\ell+k}
(1+i\tan(\theta/2)e^{-i\phi})^{\ell-k}e^{iq\phi}d\phi\
\\
\mu_{\ell,k,q}=&\sqrt{\frac{(l-k-q)!(l+k+q)!}{(l-k)!(l+k)!}}.
 \end{align*}
If $\ell  \tan(\theta/2)\ll 1$, then we can expand with respect to $\tan(\theta/2)$. Taking into account that because of integration over $\phi$ only terms 
containing $\tan^{q'}(\theta/2)$ with $q'\ge |q|$
give nonzero contribution, we obtain (\ref{as_Leg.0}) -- (\ref{as_Leg.1}).

To prove  (\ref{as_Leg}) we write
\begin{align*}
P_{00}^{(\ell)}(\cos\theta)=&\frac{1}{2\pi }\int\exp\big\{\ell u(\phi,\theta))\big\}d\phi,\quad\\
u(\phi,\theta)=&
\log (\cos(\theta/2)+i\sin(\theta/2)e^{i\phi})
+\log (\cos(\theta/2)+i\sin(\theta/2) e^{-i\phi}))\big)\\
=&\log(\cos\theta+i\sin\theta\cos\phi),\quad \\
\Re u(\phi,\theta)\le& 0,\quad \Re u(\phi,\theta)\Big|_{\phi=0\vee\pi}=0
%\\
%=&(\lambda+\lambda^{-1})\log\cos\phi +\lambda^{-1}\log (1+i\tan(\theta/2)\lambda e^{i\phi})+\lambda\log (1+i\tan(\theta/2)(\lambda^{-1} e^{-i\phi})).
\end{align*}
By (\ref{<f>}) we need to study
\begin{align*}
I_{\ell}=&\frac{W^2u_*^2\Tr S}{2}\int_{0 }^\pi\sin\theta d\theta\exp\{-4u_*^2W^2\Tr S\sin^2(\theta/4)\}P_{00}^{(\ell)}(\cos\theta)\\
%=&W^2u_*^2\Tr S\int_{-\pi}^\pi\frac{d\phi}{2\pi}\int_{0 }^\pi\sin\theta d\theta \exp\{-4u_*^2W^2\Tr S\sin^2(\theta/4)+\ell u(\phi,\theta)\}\\
=&\frac{ 
W^2u_*^2\Tr S}{2}\int_{-\pi}^\pi\frac{d\phi}{2\pi}\int_{\theta\le\tfrac{\log W}{W}}\sin\theta d\theta\exp\{-4u_*^2W^2\Tr S\sin^2(\theta/4)+\ell u(\phi,\theta)\}+O(e^{-c\log^2W}).
\end{align*}
But for $\theta\le W^{-1} \log W$ we can expand $ u(\phi,\theta)$ with respect to $\sin(\theta/2)$. We get
\begin{align}\notag
 u(\phi,\theta)=&i\varphi_{1}(\theta,\phi)-\varphi_2(\theta,\phi),\\
\varphi_{1}(\phi,\theta)=&2\sin(\theta/2)(\cos\phi+O(\sin^2(\theta/2))),\label{as_rel}\\
\varphi_{2}(\phi,\theta)=&2\sin^2(\theta/2)(\sin^2\phi+O(\sin^2(\theta/2))),
\notag\end{align}
where $\varphi_{1}(\phi,\theta)$ and $\varphi_{2}(\phi,\theta)$  are some non negative (for $\theta<\theta_0$ (with some $\theta_0$) real analytic functions.

Set $\alpha=2\ell/W$. If $\alpha\le C_0\log W$, we obtain by changing $x=2W\sin(\theta/4)$
\begin{align}\notag
I_{\ell}=&2u_*^2\Tr S \int \frac{d\phi}{2\pi}\int_0^{\infty}xdx\exp\{-u_*^2\Tr S x^2+ix\alpha\cos\phi\}+O(\alpha\log^2 W/W)\label{alpha<log}\\
=&\int \frac{d\phi}{2\pi}\hat I(\alpha\cos\phi)+O(\alpha\log^2 W/W).\notag
\end{align}
Since $\hat I(p)$ is the Fourier transform of the positive function,  there is  $\delta>0$ such that
\[
\hat I(p)<\hat I(0)-c_0p^2=1-c_0p^2,\quad |p|\le\delta,\quad \hat I(p)<1-c_0\delta^2\quad |p|>\delta,
 \]
 which implies (\ref{as_Leg}). 

If $\alpha >C_0\log W$, then we integrate by parts with respect to $\theta$ by writing 
\begin{align*}
I_{\ell}=\dfrac{
W^2u_*^2\Tr S}{2i\alpha W}&\int_{\theta\le W^{-1} \log W}\Big( \frac{d\varphi_{1}}{d\theta}\Big)^{-1} 
\frac{d}{d\theta}e^{i\alpha W\varphi_{1}(\theta,\phi)}\\
&\times \exp\{-2u_*^2\Tr SW^2 (1-\cos(\theta/2))-\alpha W\varphi_{2}(\theta,\phi)\}\sin\theta d\theta \dfrac{d\phi}{2\pi}=O(\alpha^{-1}+W^{-1}),
\end{align*}
which also clearly yields (\ref{as_Leg}). 
Here we used  that differentiation of the first term at the exponent with respect to $\theta$ gives us the $O(W)$, differentiating of $\alpha W\varphi_{2}$
gives $O(\alpha)$, 
and  by (\ref{as_rel}) 
\[ 
\Big|\frac{d\varphi_{1}}{d\theta}\Big|=\cos(\theta/2)|\cos\phi+O(\sin(\theta/2))|>C \Rightarrow\Big| \Big(\frac{d\varphi_{1}}{d\theta}\Big)^{-1}\Big|\le C';
\]
hence, the derivative of $(\frac{d\varphi_{1}}{d\theta})^{-1}$ is bounded. We recall here that  for $\alpha >C_0\log W$ with sufficiently big $C_0$
 the contribution of the integral over $\phi$ with $|\cos\phi|\le 1/2$  is $e^{-C_0\log W/2}\le W^{-1}$.

$\square$

\end{document}